\documentclass[12pt]{article}

\RequirePackage{amsthm,amsmath}
\RequirePackage{natbib}
\usepackage{mathrsfs}
\usepackage{amssymb}
\usepackage{amsbsy}
\usepackage{epsfig}
\usepackage{enumerate}
\usepackage{bm}
\usepackage{color}
\usepackage{booktabs}
\usepackage{hyperref} 
\usepackage{fullpage}
\usepackage{setspace}
\usepackage{graphicx}
\usepackage{stackrel}
\usepackage{caption}
\usepackage{subcaption}
\usepackage{multirow}
\usepackage{color}
\usepackage{verbatim}
\usepackage{bm}
\usepackage{mathrsfs}
\usepackage{bbm}
\usepackage{graphics}
\usepackage{ulem}   
\usepackage{color}  
\usepackage{float} 
\usepackage{authblk}
\theoremstyle{plain}

\newtheorem{theorem}{Theorem}[section]
\newtheorem{lemma}[theorem]{Lemma}

\newtheorem{corollary}[theorem]{Corollary}

\graphicspath{{./Figures/}}

\DeclareMathAlphabet{\mathpzc}{OT1}{pzc}{m}{it}
\allowdisplaybreaks

\graphicspath{{./figs/}}
\DeclareGraphicsExtensions{.eps,.ps,.pdf}

\newcommand{\be}{\begin{eqnarray}}
\newcommand{\ee}{\end{eqnarray}}
\newcommand{\ben}{\begin{eqnarray*}}
	\newcommand{\een}{\end{eqnarray*}}

\usepackage[utf8]{inputenc} 
\usepackage[T1]{fontenc}    
\usepackage{hyperref}       
\usepackage{url}            
\usepackage{booktabs}       
\usepackage{amsfonts}       
\usepackage{nicefrac}       
\usepackage{microtype}      
\RequirePackage{amsthm,amsmath, times, amsthm}
\RequirePackage{natbib}
\usepackage{mathrsfs}
\usepackage{amssymb}
\usepackage{amsbsy}
\usepackage{epsfig}
\usepackage{enumerate}
\usepackage{bm}
\usepackage{color}
\usepackage{booktabs}
\usepackage{graphicx}
\usepackage{stackrel}
\usepackage{caption}
\usepackage{subcaption}
\usepackage{algorithm} 
\usepackage{algpseudocode} 
\usepackage{xcolor}
\usepackage{blindtext}

\theoremstyle{plain}

\graphicspath{{./Figures/}}
\allowdisplaybreaks
\usepackage{wrapfig}
\graphicspath{{./figs/}}
\DeclareGraphicsExtensions{.eps,.ps,.pdf}

\newcommand{\bZ}{\boldsymbol{Z}}
\newcommand{\bX}{\boldsymbol{X}}

\newcommand{\bA}{\boldsymbol{A}}

\newcommand{\bbeta}{\boldsymbol{\beta}}

\newcommand{\mE}{\mathbb{E}}
\newcommand{\mP}{\mathbb{P}}
\newcommand{\mU}{\mathcal{U}}

\newcommand{\mC}{\mathcal{C}}
\newcommand{\ovar}{\operatorname{var}}
\newcommand{\ocov}{\operatorname{cov}}
\floatname{algorithm}{Procedure}
    \usepackage{array,multicol,booktabs,tabularx}
\usepackage{graphicx}

\usepackage{enumitem,linegoal}
\usepackage{calc}
\begin{document}
\title{Cluster-Adaptive Network A/B  Testing: From Randomization to Estimation}

%

\author{%
  Yifan~Zhou$^{a}$,~ Yang~Liu$^{a}$,~ Ping Li$^{b}$,~ and Feifang~ Hu$^{a}$ \thanks{The work of Feifang Hu was partially conducted while he was a consulting researcher at Baidu Research.}
			
	
	}
\date{%
    $^a$Department of Statistics, George Washington University, Washington, DC;\\ $^b$Cognitive Computing Lab, Baidu Research,  Bellevue, WA\\[2ex]%
    \today
}

\maketitle


\begin{abstract}
A/B testing is an important decision-making tool in product development for  evaluating  user engagement or satisfaction from a new service, feature or product. The goal of A/B testing is to estimate the average treatment effects (ATE) of a new change, which becomes complicated when users are interacting. When the important assumption of A/B testing, the Stable Unit Treatment Value Assumption (SUTVA), which states that each individual’s response is affected by their own treatment only, is not valid, the classical estimate of the ATE usually leads to a wrong conclusion.
In this paper, we propose a cluster-adaptive network A/B testing procedure, which involves a sequential cluster-adaptive randomization and a cluster-adjusted estimator. The cluster-adaptive randomization is employed to minimize the cluster-level Mahalanobis distance within the two treatment groups, so that the variance of the estimate of the ATE can be reduced. In addition, the cluster-adjusted estimator is used to eliminate the bias caused by network interference, resulting in a consistent estimation for the ATE. Numerical studies suggest our cluster-adaptive network A/B testing achieves consistent estimation with higher efficiency. An empirical study is conducted based on a real world network  to illustrate how our method can benefit decision-making in application.
\end{abstract}

\section{Introduction}

\noindent 

A/B testing has become the  gold standard to compare a new service, a new  strategy, or a  new feature to its old counterpart
 \cite{kohavi2013online,kohavi2014seven}.  It has been  widely used by web technology companies such as Amazon, eBay, Facebook, Google, LinkedIn, etc., to seek  data driven decisions and  the development of new products based on customers' interests.

In the context of  A/B testing,   the primary quantity of interest is usually defined as  the average treatment effect (ATE),  which can be  expressed as 
\begin{equation}
\tau = \frac{1}{N} \sum_{i=1}^{N}\mathbb{E}\left[ Y_i(Z_i=1)- Y_i(Z_i=0)\right],
\end{equation}
where $Z_i$ is the treatment assignment of the $i$ th user, i.e., $Z_i=1$ if the $i$ th  user is assigned to treatment A and $Z_i=0$ if the user is assigned to treatment B, and $Y_i(Z_i =z)$  is the response of the $i$ th user if $Z_i=z$ for $z=0,1$.  The classical estimator  (CE) uses the difference of the two sample means, i.e., 
\begin{equation}
\hat{\tau}_{CE}  = \frac{1}{N_A}\sum_{\lbrace i:Z_i=1\rbrace} Y_i(Z_i=1) - \frac{1}{N_B} \sum_{\lbrace i:Z_i=0\rbrace} Y_i(Z_i=0),
\end{equation}
where $N_A$, $N_B$ are the number of  users in treatment A and B respectively. In practice,  the experimenter will first assign users with  complete randomization, i.e.,  assign users to treatments A or B with  equal probabilities of 1/2, and then use $\hat{\tau}_{CE}$ to estimate the ATE. This method relies on  the  \textit{ Stable Unit Treatment Value Assumption} (SUTVA) \cite{imbens2015causal}, which requires that the response of each user in the experiment depends only on their own treatment and not on the treatments of others. Under the SUTVA, $\hat{\tau}_{CE}$ is unbiased and can provide reliable estimation of the ATE.  However,  users  in an A/B testing are commonly involved in  a social network, so that a user's behavior/response can be influenced by his/her social neighbors/friends. For example, if a user frequently uses an emoji when sending messages, then his/her friends may be affected to use the same emoji.  This phenomenon is typically called  \textit{network effects}, also known as \textit{peer influence} or \textit{social interference} \cite{eckles2016design}. In this case, $\hat{\tau}_{CE}$ is biased and can not provide valid estimation of the ATE \cite{eckles2016design,Gui2015}.  

Furthermore, it is often that the  users in a social network can be  grouped into different clusters with different features. The users in the different clusters usually have different behaviors with their  responses, whereas the users within the same clusters  share similar behaviors. For example,  users within the  same age group may have the same taste for music or like the same kind of movies, while the users within different age groups would have different preferences on music, movies, etc. We refer to this phenomenon as \textit{cluster effect}.


To demonstrate the effects caused by a social network for A/B testing, consider a network graph $\mathcal{G}=(V,E)$  with $N= |V|$ vertices (users).  Suppose the  vertices of $V$ can be partitioned into $m$ disjoint clusters $\mC_1,...,\mC_{m}$,  and   $n_j$ is  the cluster size of $\mC_j$.  Let $\bm{X}_j$ be the $p$-vector of covariates  of the $j$ th cluster, of which the covariates represent some important common features of the cluster, i.e., cluster size, number of inside-edges,  education levels of the users,  age groups of the users,  etc. If the $i$ th user is in $\mC_j$, the assumed  model  for  the $i$ th user's response is  
\begin{equation}
\label{Y.eq} Y_i  =    \mu_0(1-Z_i) + \mu_1 Z_i + \alpha_0 \bA_{i*}\bZ (Z_i -1) + \alpha_1 \bA_{i*}(\bZ-\boldsymbol{1}) Z_i + \bbeta^{\text{t}}\bX_j + \epsilon_i, 
\end{equation}
where $\mu_1,\mu_0$ are the effects  of treatments A and B, respectively,  $\bbeta$ is  the cluster effect, and $\epsilon_i\sim N(0, \sigma^2_\epsilon)$ is i.i.d. random error. Here,   $\bA_{i*}$ is the $i$-th row of the adjacency matrix $\bA$. In addition,    $\bZ = (Z_1,...,Z_{N})^{\text{t}}$ is the  vector of assignments for all $N$ users and $\boldsymbol{1} = (1,...,1)^{\text{t}}$  is a vector of ones. Therefore, $\alpha_0, \alpha_1$ represent the network spill-over effect, i.e., if the  effect on the $i$ th user if he/she is in treatment A/B but has neighbors in treatment B/A, respectively. In literature, different authors assume different response models to demonstrate the effect caused from a social network \cite{Gui2015,jiang2016,middleton2011unbiased,raudenbush1997statistical,saint2019using,Ugander2013}. In this paper, we consider the response model (\ref{Y.eq}), because it can help us to study both the cluster effect and the network spill-over effect. Moreover, this model can be used to study  a variety of real world networks. Thus, (\ref{Y.eq}) is generous enough for the illustrating purpose.

From (\ref{Y.eq}), it can be seen that if the  network spill-over effect does not exist, i.e.,  $\alpha_0=\alpha_1=0$,  $\hat{\tau}_{CE}$ is an unbiased estimator of  $\tau= \mu_0-\mu_1$, but $\hat{\tau}_{CE} \not \to\tau$ when $\alpha_0,\alpha_1\ne0$. Therefore, the network effect should be  taken into account for the randomization step as well as the estimation step in the A/B testing problem. To eliminate the network effect, recent works suggest that  using  clusters as a unit for randomization following with an adjusted estimator can improve the estimation  of the ATE \cite{Gui2015,middleton2011unbiased,raudenbush1997statistical,Ugander2013}.  These proposed  procedures usually consist of  three steps: 1. community detection, 2. randomization over clusters, and 3. estimation of the ATE based on an assumed response model.  These methods are not "comparable",  because different procedures consist of different community detection methods, different randomization methods, and are based on different response models. Especially, different community detection methods result in very different results for the estimation of the ATE, even if  the last two steps use the same kind of approaches. For the randomization step, the most commonly used procedure is the complete randomization of  clusters, which does not take the  features of the clusters into account. Very few studies have considered  using the cluster's feature in randomization to produce better estimation results.

Unlike the design for network A/B testing, in the design of clinical trials and causal inference,  various  procedures are proposed to assign the users' treatments by   using the information of the users' features based on different criterion  \cite{hu2012asymptotic,morgan2012rerandomization,pocock1975sequential,qin2016pairwise,taves1974minimization}.  Recent research \cite{bugni2018inference,ma2015testing,ma2019statistical,shao2010theory} has shown that the randomization procedure which utilizes the covariates of the users to produce balanced treatment arms can improve the efficiency of the estimation of the ATE.  However, how to develop  a randomization  procedure using this idea in the context of network A/B testing is  still an open problem. 

In this paper,  we  solve this problem by showing  that using an appropriate randomization procedure which utilizes the information of the clusters and by  following with an appropriate estimator could not only obtain valid  (consistent) but also an efficient (reduced variance)  estimation  of the ATE. We propose a cluster-adaptive network testing procedure consisting  of two steps: first, a  cluster-adaptive randomization (CAR) procedure which sequentially assigns the clusters  to minimize the Mahalanobis distance of the cluster's covariates  between  treatment arms,  then   follow with  a cluster-adjusted estimator (CAE) which is consistent for the estimation of the ATE under randomization of clusters. To make the results caused from our randomization procedure comparable with others, we assume that the underlying network is known and fixed  throughout this paper. In practice, one should obtain the graph of the network first. This can be achieved by different community detection algorithms  \cite{bader2013graph,newman2006modularity,Ugander2013}.  Theoretical results as well as numerical studies provide the justification for the superiority of our procedure. 

The outline of the paper is as follows. In Section \ref{Section: CAN A/B testing}, we propose the  cluster-adaptive network A/B testing procedure and provide the theoretical guarantee for our procedure. Numerical  studies  of a hypothetical network  and the  MIT phone call network \cite{datarepository} are presented in Sections \ref{Section: Hypothetical network} and \ref{Section: Real world Network}. Finally,  a conclusion is drawn in Section \ref{Section:Conclusion}. Detailed proofs of the theoretical results, as well as  comprehensive numerical studies are presented in the supplementary material.


\section{Cluster-Adaptive Network A/B Testing}
\label{Section: CAN A/B testing}
\noindent

Suppose we have already achieved the partition of  the network graph,  so that the network is known.  To improve the efficiency for the estimation of the ATE, we consider  balancing the cluster features via a CAR procedure proposed in Section \ref{Section: Cluster-adaptive randomization}. Then we introduce a  CAE    in Section \ref{Section: Cluster-adjusted estimator}  to  obtain a consistent estimation of the ATE. Theoretical guarantees are provided   in Section \ref{Section: Thoeretical results} to justify our cluster-adaptive network A/B testing procedure.
\vspace{-0.1in}

\subsection{Cluster-Adaptive Randomization (CAR)}
\label{Section: Cluster-adaptive randomization}

\noindent

	\begin{algorithm}[H]
	\caption{Cluster-Adaptive Randomization (CAR)} 
	\begin{algorithmic}[1]
		\State \textbf{Input: }$\bX_1, ..., \bX_m$; $1/2<q<1;$
	\State Calculate $\operatorname{cov}(\bX)^{-1}$ based on $\bX_1,\cdots,\bX_m$

	\For {$k= 1$ to $\lceil m/2 \rceil$}
	\If {$2k \leq m$}  
	\State Let  $M_{2k}^{(1)}$ and  $M_{2k}^{(2)}$  be the pseudo Mahalanobis distance the first $k^{th}$  pairs and 
	\State calculated $M_{2k}^{(1)}$ from (\ref{Ma.dist}) given $(T_{2k-1}, T_{2k}) = (1, 0)$ 
	\State calculated $M_{2k}^{(2)}$ from (\ref{Ma.dist}) given $(T_{2k-1}, T_{2k}) = (0, 1)$
	\State $(T_{2k-1}, T_{2k}) \leftarrow (0, 1)$
	\If {$M_{2k}^{(1)} < M_{2k}^{(2)}$}
	\State Replace $(T_{2k-1}, T_{2k})$ with $(1, 0)$ with probability $q$
	\Else
	\If {$M_{2k}^{(1)} > M_{2k}^{(2)}$}
	\State Replace $(T_{2k-1}, T_{2k})$ with $(1, 0)$ with probability $1 - q$
	\Else
	\State Replace $(T_{2k-1}, T_{2k})$ with $(1, 0)$ with probability $1/2$
	\EndIf
	\EndIf
	\Else 
	\State $T_{2k-1} \leftarrow 1$ with probability 0.5, and 0 otherwise
	\EndIf
	\EndFor
	\end{algorithmic} 
	\label{Cluster-Adaptive Randomization Procedure}
\end{algorithm}

We adopt the idea similar to the ones used in  \cite{morgan2012rerandomization,qin2016pairwise}, which consider the  Mahalanobis distance as the imbalance measure.  Let  the Mahalanobis distance of  the first $2k$  assigned clusters be
 \begin{align}
 M_{2k} & = \left(\bar{\bm{X}}_{A}-\bar{\bm{X}}_{B}\right)^{\text{t}}\operatorname{cov}\left(\bar{\bX}_{A}-\bar{\bX}_B\right) ^{-1}\left(\bar{\bm{X}}_{A}-\bar{\bm{X}}_{B}\right) \nonumber\\
 &\propto \frac{k}{2} \cdot \left(\bar{\bm{X}}_{A}-\bar{\bm{X}}_{B}\right)^{\text{t}}\operatorname{cov}(\bX) ^{-1}\left(\bar{\bm{X}}_{A}-\bar{\bm{X}}_{B}\right), \label{Ma.dist}
 \end{align}  
where   $\bar{\bm{X}}_A$, $\bar{\bm{X}}_B$ are the $p$-vector  sample means of the covariates of clusters for  treatments A and B calculated with the first $2k$ clusters, and $\bm{X}$ is the  $m\times p$ matrix of the clusters' covariates for all $m$ clusters. To introduce our CAR procedure,  let $T_j$ be the treatment assignment for the $j$ th cluster, so that if $T_j =1$ then $Z_i=1$ for all  $i\in\mathcal{C}_j$,  and if $T_j =0$ then  $Z_i =0$ for all $i \in\mC_j$. Suppose   $\bm{X}_1,...,\bm{X}_m$ are observed before the randomization. We present  CAR    in  Procedure \ref{Cluster-Adaptive Randomization Procedure}.

The idea of  CAR is as follows.  CAR  sequentially assigns a pair of clusters, so that only one of the two clusters is assigned to treatment A. Resulting from the pairwise sequential treatment assignments, there are only  two possible outcomes for the assignment and  two possible values for the  Mahalanobis distance, i.e., $M_{2k}^{(1)}$, and $M_{2k}^{(2)}$. The procedure assigns   a higher probability $q$ to the assignment  which leads to the smaller value of the Mahalanobis distance. For example,  $q=0.85$ is used through out this  paper. The covariates' imbalance of the clusters is then minimized for a large value of  $m$.   As  balanced  treatment arms of clusters are generated by CAR,  this  randomization procedure can improve the efficiency of estimation, as will be shown in Section \ref{Section: Thoeretical results}.

\subsection{Cluster-Adjusted Estimator (CAE)}

\label{Section: Cluster-adjusted estimator}

\noindent

Once the CAR has been done, each of the users in the same cluster will have the same treatment assignment. Therefore, the network  effect obtained from the users in the same cluster is removed.  To eliminate the network effect obtained from linked clusters, define  $\mU$ as the uncontaminated set of vertices, that is 
$$\mU =  \{i: \bA_{i*} (\boldsymbol{1}-\bZ) = 0 \text{ and } Z_i =1\} \cup \{i: \bA_{i*} \bZ = 0  \text{ and } Z_i =0\}, \quad i= 1,...,n.$$
Then the adjusted estimator taking the  sample averages by using  the users in $\mU$ should be consistent. However, this estimator does not take the cluster effect into account. If the cluster effect  is not adjusted appropriately, the estimation of the ATE can still be affected. For example, consider a network with several clusters, where one cluster has  the largest cluster size and  all other clusters are about the same size. In this case,  the assignment of the largest cluster can affect the estimation of the ATE, due to the inflation of the cluster effect of this cluster. The intuition behind this is  that  the largest cluster will have more weight than the other clusters, if  the sample average is used.

To introduce our adjusted estimator, let  $n_j(\mathcal{U}) \ne0$ and $\bar{Y}_j(\mathcal{U})$ be the number of  uncontaminated users and the sample average of the uncontaminated users in  the $j$ th cluster, respectively. The cluster-adjusted estimator (CAE) is proposed   as  follows, 
\begin{align}
	\hat{\tau}_{CAE} &= \frac{1}{m_A}\cdot \sum_{j=1}^{m}T_j \bar{Y}_j(\mathcal{U})- \frac{1}{m_B}\cdot\sum_{j=1}^m (1-T_j) \bar{Y}_j(\mathcal{U}), \label{CAE}
 \end{align}
where  $m_A$,  $m_B $ are the number of clusters assigned to treatment A and B, respectively. It can be seen from (\ref{CAE}) that the cluster effect will not affect the estimation of the ATE, because the averages taken within each of the clusters remove the inflation of the cluster effect caused by unequaled cluster sizes. 


\subsection{Theoretical Properties}
\label{Section: Thoeretical results}

\noindent

To show the good properties of   our cluster-adaptive network A/B testing procedure, we first present  the balance properties of CAR in  Theorem \ref{Th2}.   

\begin{theorem}[Balance property of Cluster-Adaptive Randomization] 
	\label{Th2}
	Let $M_{m}$ be the Mahalanobis distance after the assignment of  $m$ clusters under CAR, then 
\begin{equation}
\label{th2.eq1}
	\begin{array}{cccc}
	M_{m} = O_p\left(\frac{1}{m}\right), & 
	and 
	& M_{m} \xrightarrow{\mathscr{P}} 0, &\text{ as } m \rightarrow \infty.
	\end{array}
	\end{equation}
\end{theorem}

The proof of Theorem \ref{Th2} utilizes the drift condition, i.e., see Chapter 11 of \cite{meyn2012markov}, to verify that $\lbrace kM_{k}\rbrace_{k=1}^{\infty}$ is a Harris recurrent markov chain. This goal can be achieved by showing that the increment of $M_{m}$ in each step of the pairwise treatment assignments  in CAR is bounded by some positive constant. We next study the consistency of CAE in the  next theorem.

	\begin{theorem}[Consistency of Cluster-Adjusted Estimator] 
		\label{Th1}
	Under any  randomization performed  on clusters,  the cluster-adjusted estimator eliminates the network effect, i.e., 
	\begin{align}
		\hat{\tau}_{CAE}& = \mu_1-\mu_0 + \frac{1}{m_A}\sum_{j=1}^m T_j\bar{W}_j (\mathcal{U})- \frac{1}{m_B} \sum_{j=1}^{m}(1-T_j)\bar{W}_j(\mathcal{U}), \label{th1.presentation}\\
			where \ \ \bar{W}_j(\mathcal{U}) & =  \frac{1}{n_j(\mathcal{U})} \sum_{i\in \mathcal{C}_j\cap \mathcal{U}} (\bm{X}^{\text{t}}\bm{\beta}+ \epsilon_i). \nonumber
	\end{align}
Furthermore, if the employed  randomization procedure  satisfies the following balance condition,
\begin{equation}
\label{ass1}
\begin{array}{cccc}
\frac{1}{m} \sum_{j=1}^{m}T_j \stackrel{\mathscr{P}}{\longrightarrow}\frac{1}{2},& and & \frac{1}{m}\sum_{j=1}^{m}(2T_j-1)\bm{X}_j\stackrel{\mathscr{P}}{\longrightarrow}0, & as \ m \to \infty,
\end{array}
\end{equation} then the  cluster-adjusted estimator is consistent,  that is 
\begin{align*}
\begin{array}{cc}
\hat{\tau}_{CAE}\stackrel{\mathscr{P}}{\longrightarrow} \tau, & as\ m \to \infty.
\end{array}
\end{align*}
\end{theorem}

Theorem \ref{Th1} shows that the cluster-adjusted estimator can not only eliminate the network effect, but is consistent if the balance condition (\ref{ass1}) is satisfied. 

\begin{corollary}
\begin{enumerate}[leftmargin=*]
\item Under the graph cluster randomization \cite{Ugander2013}, i.e., performing the complete randomization on clusters (CRC), equation (\ref{ass1}) is satisfied, then  CAE is consistent.
    
 \item  Under CAR, equation (\ref{ass1}) is satisfied, then  CAE is consistent. 
\end{enumerate}
\label{Cor1}
\end{corollary}

Corollary \ref{Cor1} shows that CAE is consistent for CRC and CAR. Under CRC, each cluster is assigned to treatment A with probability 1/2. As no information of $\bX_j$ is used for randomization, $T_j$ is independent with $\bm{X}_j$. Therefore,  it is easy to verify that (\ref{ass1}) for CRC.  For CAR, (\ref{ass1}) follows from  Theorem \ref{Th2} that $\bar{\bX}_A-\bar{\bX}_B$ converges to zero in probability, as $m\to \infty$ and Slutsky's theorem.  It is obvious that CAE is still not consistent for complete randomization with users (CRU). The reason is that this user-level randomization procedure cannot remove the network spill-over effect obtained from the users within the  same cluster.

To compare the two consistent network A/B testing procedures, i.e., using CRC following with CAE, and using CAR following with CAE,  we introduce the  percent reduction in variance (PRIV) of an estimator $\hat{\tau}$ to  compare CAR with CRC, 
\begin{equation}
PRIV(\hat\tau|CAR) = \frac{\ovar[\hat\tau|CRC] - \ovar[\hat\tau|CAR]}{\ovar[\hat\tau|CRC]}. 
\end{equation}
It is obvious that if $PRIV(\hat\tau|CAR)>0$ then the new procedure is better than CRC. The PRIV of the  CAE,  comparing the randomization procedures CAR and  CRC, is studied in the next theorem. 

\begin{theorem}
	\label{Th3}
If the  cluster effect exists ($\bbeta\ne 0$),  then
\begin{align}
PRIV(\hat\tau_{CAE}|CAR) & \geq \left(1-\frac{\mathbb{E}[M_m|CAR]}{p}\right)R^2_C\to R^2_C ~\text{ as } m \rightarrow \infty, \label{lb.eq}
\end{align}
		where $R_C^2$ is the squared multiple correlation between $Y^*_j$ and $\bX_j$ within the treatment group. Here,  $Y^*_j$ is  a pseudo response generated by the model without network effect,  i.e.,  
		\begin{equation}
		    \begin{array}{ccc}
		    	Y^C_j = \mu_0(1-T_j) + \mu_1 T_j + \bbeta' \bX_j + \epsilon_j, &for& j= 1,...,m.
		    \end{array} \label{psedo}
		\end{equation}
	
\end{theorem}    
The positive lower bound in (\ref{lb.eq}) is attained when each of the cluster has only one uncontaminated user.  In addition, this lower bound converges to $R_C^2$ as the  number of clusters goes to infinity. Therefore, the $PRIV(\hat{\tau}_{CAE}|CAR)$ can achieve  the maximum value of this lower bound and greatly  improve the  estimation efficiency. Furthermore, $R^2_C$  increases  as more covariates are included in the  pseudo response  model \cite{rao1973linear}. As a result, the efficiency of using CAR can be improved by using more covariates of clusters for randomization.

\section{The Hypothetical Network}
\label{Section: Hypothetical network}

\noindent

In this section we use a hypothetical network to show the advantages of our cluster-adaptive network A/B testing procedure.The simulated  network is generated as follows. We first use  the  Watts-Strogatz small-world model to generate 500 independent clusters  with no interactions. For $j = 1,\cdots,500$, we generate  the cluster size $n_j$ from a symmetric discrete distribution,  $\mP(n_j = \lambda) = \left[\left(|\lambda- 20|+ 0.5\right)\sum_{\lambda=10}^{30}(1/ \left(|\lambda- 20|+ 0.5\right)\right]^{-1}$, where $\lambda=10,...,30$. Therefore, $N = \sum_{j=1}^{500}n_j$ users are involved in this network.  In addition, $r\times N$ edges are randomly added, so that  several clusters are connected, where  $r\in\lbrace 0.1,0.5 \rbrace$ is chosen as the re-connection probability.

Consider the response $Y_i$ follows the assumed response model (\ref{Y.eq}) with $p=4$ and $\epsilon_i\sim N(0, 2^2)$.  The covariates of the cluster,  i.e., $\bX_j= (X_{j1},\cdots,X_{j4})^{\text{t}}$, involved in the response model are  the  number of vertices in $\mC_j$ ($X_{j1}$), the number of edges in $\mC_j$ ($X_{j2}$), the  number of edges in $\mC_j$ with other clusters ($X_{j3}$), and the density of $\mC_j$ ($X_{j4}$). In addition, consider $\alpha_1=-\alpha_0=\alpha \in \lbrace 0,0.25,0.5,0.75,1\rbrace$ for the network spill-over effect, $\beta_1= \cdots= \beta_4= 1$ for the effect of cluster covariates, and $\tau=\mu_1-\mu_0 = 1$  for the treatment effect.




To compare with other network A/B testing procedures, we consider  the classic A/B testing procedure, i.e., complete randomization with users (CRU),   following with the classical estimator (CE), and the  graph  clustering procedure using the  complete randomization on clusters (CRC), following with CE and CAE. For our CAR method, we consider two cases,  1)  CAR (2),  the CAR using the first  two  covariates   $X_{j1}$ and $X_{j2}$, and 2) CAR (4), the CAR using all of the four covariates $X_{j1},\cdots,X_{j4}$. Following the two scenarios of CAR, we consider both CE and CAE for estimation.  For the  simplicity of the presentation,  we present some summary statistics comparing the balance properties of the two cluster - level randomizations in Figure \ref{Figure 0: balance properties}.    We summarize the bias  and standard deviation of the estimated ATE under the settings of  $\alpha\in \lbrace 0,1\rbrace$, and $r\in \lbrace 0.1,0.5\rbrace$ in Table \ref{Tabel 1: bias and Variance of Hypothetical network}.  The $PRIV(\hat{\tau}_{CAE}|CAR)$ is calculated in Table \ref{Table 2: PRIV for the hypothetical}. Other results exploring different parameter settings are presented in Figures \ref{Figure 1: Bias and s.d. of  the Hypothetical network}. More detailed simulation results  can be found in the supplementary material.

\begin{figure}[H]
	\centering
	\includegraphics[width=1\linewidth]{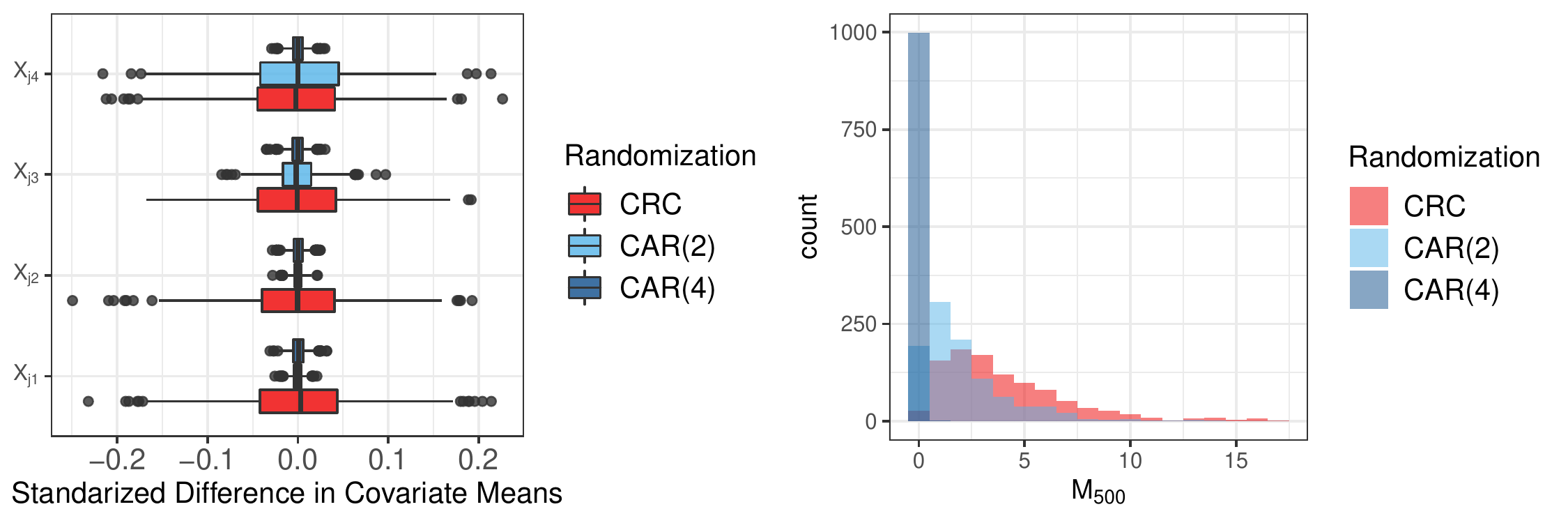}
		\caption{The standardized difference in means for  $\bX_j$ and the Mahalanobis distance when $\alpha=0.5$, $r=0.1$, given the hypothetical network, based on 1000 runs. The Mahalanobis distance is calculated for all four covariates of the clusters.   }
			\label{Figure 0: balance properties}
\end{figure}

\begin{figure}[H]
	\centering
	\includegraphics[width=1\linewidth]{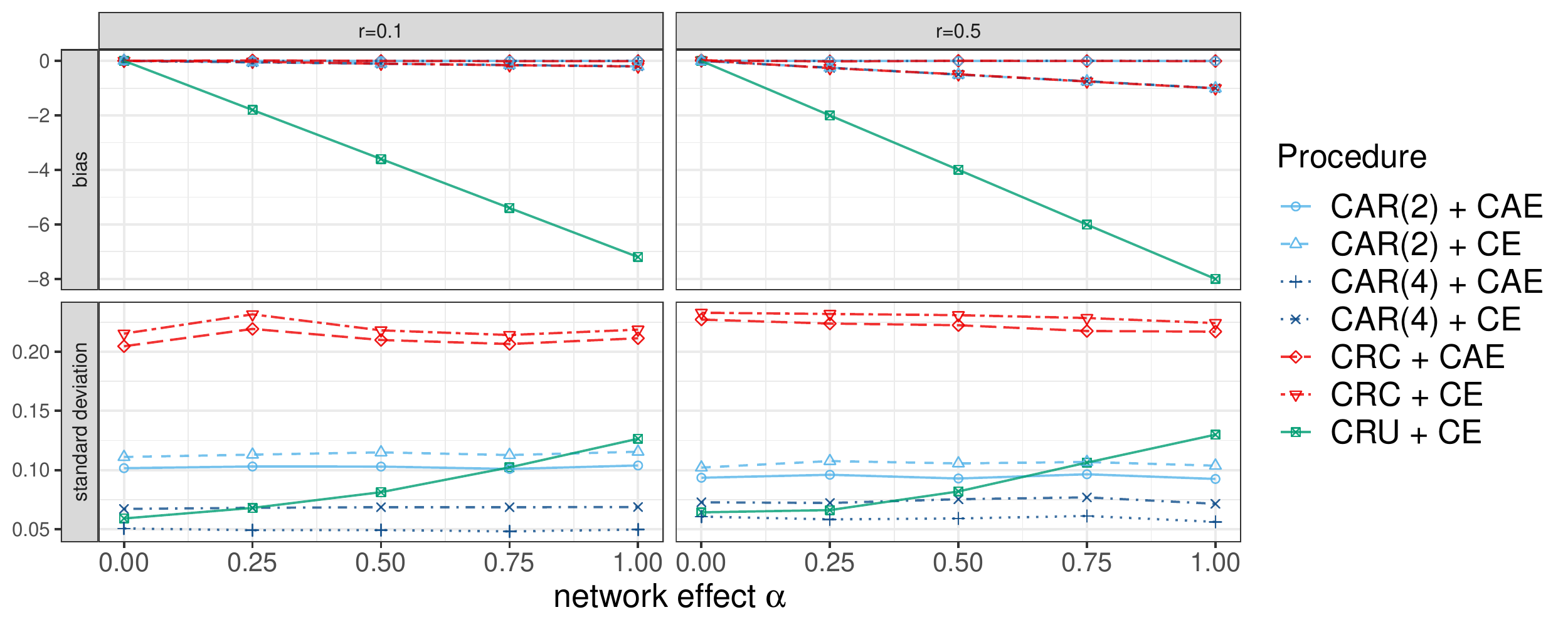}
		\caption{The bias and standard deviation of the estimated ATE versus network effect $\alpha$ given the hypothetical network, based on 1000 runs.}
			\label{Figure 1: Bias and s.d. of  the Hypothetical network}
			\vspace{-0.15in}
\end{figure}

\vspace{-0.1in}

\begin{table}[H]

	\centering
	\caption{The bias and standard deviation of the estimated of ATE under different A/B testing procedures given the hypothetical network,  based on 1000 runs.}
		\label{Tabel 1: bias and Variance of Hypothetical network}
				\begin{tabular}{cclcrrrrr}
			\hline
\multirow{2}{*}{$r$}&   \textit{Randomization} & \multirow{2}{*}{\textit{Estimator}  }  && \multicolumn{2}{c}{$\alpha = 0$} & & \multicolumn{2}{c}{$\alpha = 1$}  \\
			\cline{5-6} \cline{8-9}
			& \textit{procedure} &  && \textit{Bias} & \textit{s.d.}& & \textit{Bias} & \textit{s.d.} \\ 
		\cline{1-3}		\cline{5-6} \cline{8-9}
			\multirow{10}{*}{$0.1$} &\multirow{2}{*}{CAR ($2$)} & CAE   &     & 0.0057  & 0.10 && -0.0003 & 0.10 \\
	                               		&               & CE   &      & 0.0047  & 0.11 & & -0.2005 & 0.12 \\
	                               \cline{2-3}	\cline{5-6} \cline{8-9}
	&\multirow{2}{*}{CAR ($4$) } & CAE&   		& -0.0022 & 0.05 & & 0.0025  & 0.05 \\
				&           & CE   	&	& -0.0012 & 0.07 & & -0.1976 & 0.07 \\
			 \cline{2-3}	\cline{5-6} \cline{8-9}
			& \multirow{2}{*}{CRC} & CAE     &   & 0.0002  & 0.21 & & 0.0027  & 0.21 \\
			&               & CE   &     & 0.0012  & 0.22 &  & -0.1959 & 0.22 \\
		 \cline{2-3}	\cline{5-6} \cline{8-9}
			& CRU           & CE    &     & 0.0015  & 0.06 & & -7.2003 & 0.12 \\
		\cline{1-3}		\cline{5-6} \cline{8-9}
			\multirow{10}{*}{$0.5$}  &\multirow{2}{*}{CAR ($ 2$)}   &CAE &    & -0.0022 & 0.09 & & -0.0047 & 0.09\\ 
			&               & CE    &    & -0.0046 & 0.10 & & -1.0054 & 0.10 \\
		 \cline{2-3}	\cline{5-6} \cline{8-9}
&\multirow{2}{*}{CAR ($ 4$)}  & CAE   	&	& 0.0027  & 0.06&  & 0.0002  & 0.06 \\
			&               & CE   	&	& -0.0003 & 0.07 & & -1.0019 & 0.07 \\
		 \cline{2-3}	\cline{5-6} \cline{8-9}
			& \multirow{2}{*}{CRC} & CAE    &    & 0.0019  & 0.22 & & -0.0011 & 0.22 \\
			&               & CE   &      & -0.0002 & 0.23 & & -1.0033 & 0.23  \\
		 \cline{2-3}	\cline{5-6} \cline{8-9}
			& CRU           & CE    &     & -0.0003 & 0.06& & -7.9999 & 0.12 \\
			\hline
		\end{tabular}
\end{table}

\textbf{Results:} Figure \ref{Figure 0: balance properties} shows that the CAR produce more balanced clusters than CRC. The Mahalanobis distance calculated with all four covariates of the clusters is more concentrated at zero, when more covariates are used in CAR. In addition, it can be seen from the boxplot that use of the covariates in CAR will result in a small variance of a standardized difference of the covariates mean. These results show the superior balance properties of the  CAR procedure.

 \begin{wrapfigure}[9]{r}{3in} 
 \vspace{-0.3in}
\begin{minipage}[b]{0.95\linewidth}
	\begin{table}[H]
	\centering
	\caption{The $PRIV(\hat{\tau}_{CAE}|CAR)$ given  the  hypothetical network, based on 1000 runs. }
	\label{Table 2: PRIV for the hypothetical}
		\begin{tabular}{cccrrrrrrr}
	\hline
$r$	& CAR& &$\alpha=0$ & $\alpha= 1$\\
\cline{1-2} \cline{4-5}
	\multirow{2}{*}{0.1}& (2)&&75.67  & 74.93 \\
	                                             &(4)&& 93.86  & 94.46 \\ 
\cline{1-2} \cline{4-5}
	\multirow{2}{*}{0.5}	& (2)& & 82.67  & 82.95  \\
	          & (4)& &92.86  & 93.27  \\
	\hline
\end{tabular}
\end{table}

\end{minipage}

\end{wrapfigure}

It can be seen from Table \ref{Tabel 1: bias and Variance of Hypothetical network} and Figure \ref{Figure 1: Bias and s.d. of  the Hypothetical network} that all of the A/B testing procedures are consistent, i.e., bias is about zero, when there is no network effect. In this case, the classical A/B testing procedure works well.  As the network effect increases, the bias using the classical A/B testing procedure increases dramatically, i.e., the bias under CRU following with CE is above 7, when $\alpha=1$.

The network A/B testing procedures  using CE, i.e., CAR following with CE and CRC following with CE, are less biased than the classical method. Comparing the bias obtained from these procedures for  $r=0.1$ and $r=0.5$, it can be seen that these methods work well when the clusters are less connected, i.e., $r=0.1$. When the connections between the clusters increase, the biases of these procedures are not negligible. For example, when $r=0.5$, the bias of CAR following with CE is -0.1976. This is  because  the effects obtained from linked clusters, which can bias the estimation, are not removed if CE is used for estimation. On the other hand, the network A/B testing procedure using CAE,  i.e., CAR following with CAE and CRC using CAE, is consistent for all of the parameter settings.  These results justify Theorem \ref{Th1}  and suggest that if the network effect is not eliminated appropriately, the bias will still increase as the network effect increases.

From Table \ref{Tabel 1: bias and Variance of Hypothetical network} and Figure \ref{Figure 1: Bias and s.d. of  the Hypothetical network}, it can be seen that using  CAR  procedure can significantly reduce the standard deviation of the estimators of the ATE.  Comparing the performance of  CAR (2) and CAR (4), the  standard deviation decreases  as more covariates of the clusters are used in CAR. This suggests that the efficiency of the estimation will increase if more important features of the clusters can be used in  CAR.  Furthermore, Table \ref{Table 2: PRIV for the hypothetical} shows that the improvement of using CAR is above $70\%$. In particular, using all four covariates in CAR can improve at least $93\%$ of the variance for the estimation of ATE.  These results  suggest that using our  cluster-adaptive network A/B testing procedure is able to  not only reduce the bias but also improve the estimation efficiency.  

\section{The MIT Phone Call Network}
\label{Section: Real world Network}
\noindent

 	The MIT Phone Call Network listed in the Network Data Repository   \cite{datarepository} is used to illustrate how our cluster-adaptive network A/B testing procedure works for real world social networks. 
	The network consists of phone calls/voicemails between a group of users at MIT, where vertices and edges represent users and calls/voicemails, respectively \cite{eagle2006reality}.  In addition,  the network is labeled with 82 clusters via label propagation algorithm \cite{PhysRevE2007labelpropagation} (see Figure \ref{Figure 3: MIT Phone call network}).  The response of  users is generated with the same parameter  settings as  in Section \ref{Section: Hypothetical network}.  The  A/B testing procedures used in Section \ref{Section: Hypothetical network} are also compared for the MIT phone call network.

\begin{wrapfigure}{r}{0.42\textwidth}
	\vspace{-0.4in}
	\begin{center}
		\includegraphics[width=1\linewidth]{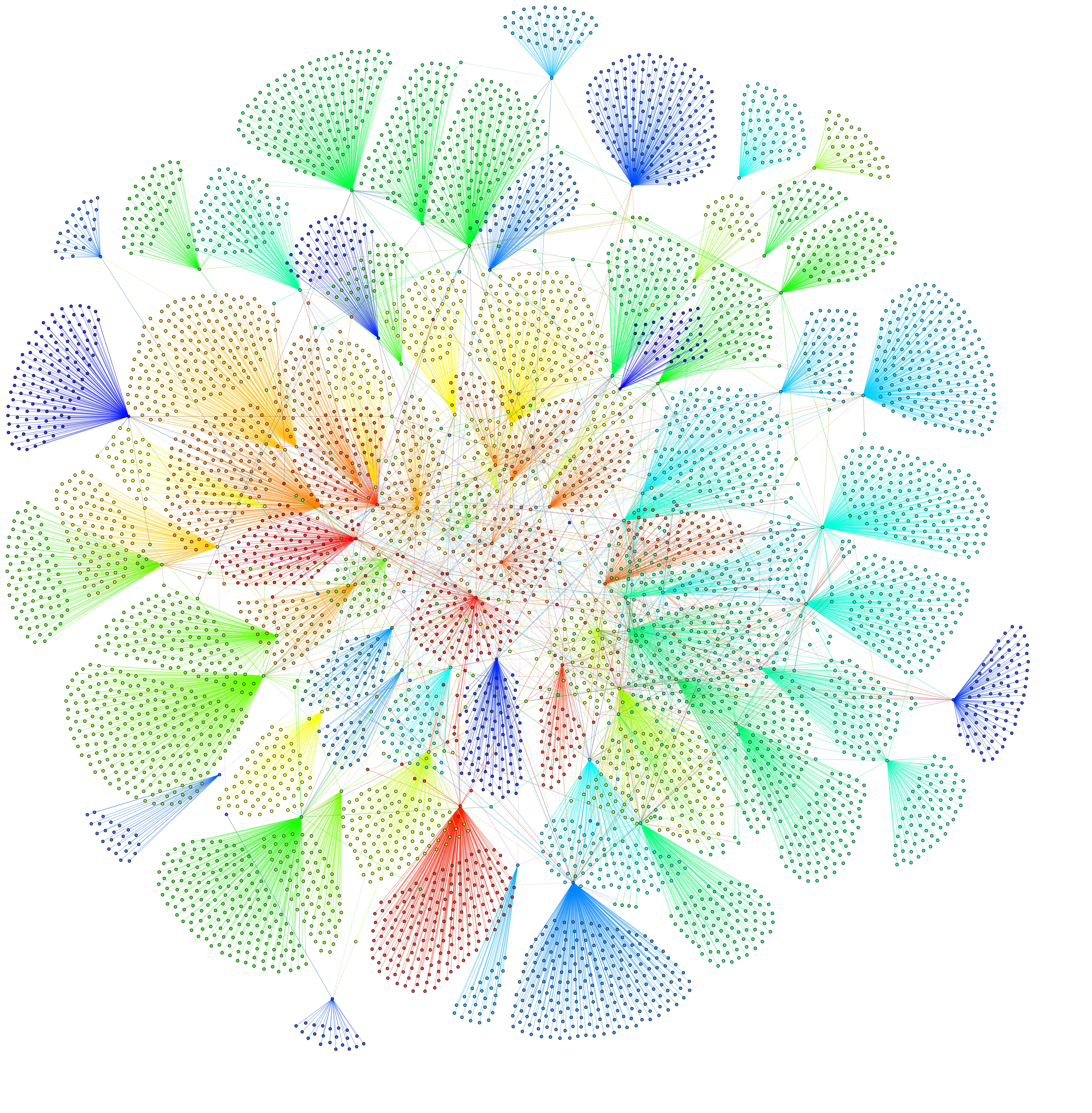}
	\end{center}
	\vspace{-0.2in}
	\caption{MIT phone call network: There are 6819 vertices and 7768 edges, which are colored by labels of  the~clusters. }
		\vspace{-0.2in}
\label{Figure 3: MIT Phone call network}

\end{wrapfigure}

For the numerical comparison, we consider that all four covariates of the cluster are used for CAR. 
The bias  and standard deviation  of the estimated ATE under the setting of  $\alpha \in \lbrace 0,1\rbrace$  are presented  in Table \ref{Tabel 3: bias and Variance of  MIT phone call network}. The simulation results exploring other  parameter settings are presented in Figures \ref{Figure 4 : Performance MIT phone call network}.

\textbf{Results:}  Similar to the results presented in the previous section, the classical A/B testing procedure works well when there is no network effect. However,  bias caused by using this method increases dramatically as the network effect increases.  Comparing with the values under CRU,  CE  under CAR and CRC are moderately biased. CAE under CAR and CRC are consistent for the estimation of the  ATE, so that the biases are about zero. It can be seen from Figure \ref{Figure 4 : Performance MIT phone call network} and Table \ref{Tabel 3: bias and Variance of  MIT phone call network} that our cluster-adaptive network A/B testing procedure has much smaller standard deviation compared to CRC following with CAE. According to our simulations, $PRIV(\hat\tau_{CAE}|CAR)$ under $\alpha=0$ and $\alpha=1$ are $84.52\%$  and  $84.55\%$, respectively. These results  suggest that using our  cluster-adaptive network A/B testing procedure is able to  not only reduce the bias caused by the MIT phone call network but also improve efficiency for  estimation of the ATE.

\vspace{-0.1in}
\begin{figure}[H]
	\centering
			\includegraphics[width=1\linewidth]{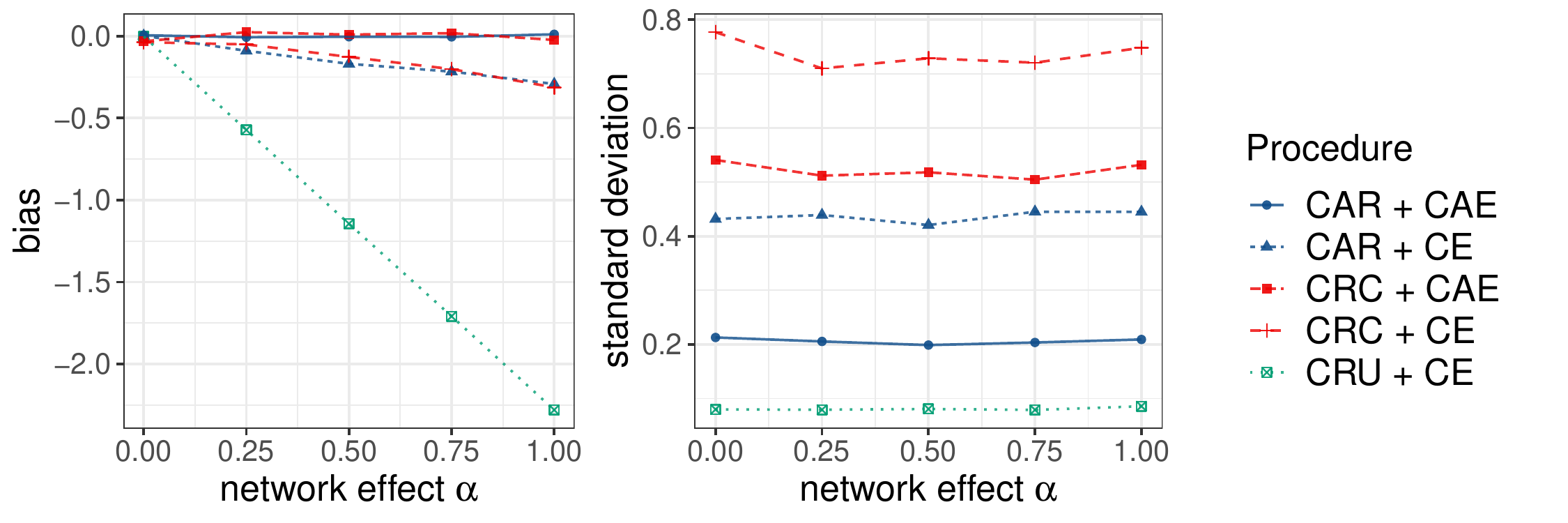}
\caption{The bias and standard deviation of the estimated ATE versus network effect $\alpha$ given  the MIT phone call network, based on 1000 runs.}
\label{Figure 4 : Performance MIT phone call network}
\end{figure}

\vspace{-0.1in}

	\begin{table}[H]
\vspace{-0.2in}
	\centering
	\caption{The bias and standard deviation  of  the estimated ATE  given the MIT phone call network, based on 1000 runs. }
	\label{Tabel 3: bias and Variance of  MIT phone call network}	
		\begin{tabular}{clrrrrrrrr}
		\hline
			\textit{Randomization} & \multirow{2}{*}{\textit{Estimator}} & & \multicolumn{2}{c}{$\alpha = 0$} & & \multicolumn{2}{c}{$\alpha = 1$} \\
			\cline{4-5} \cline{7-8}
		\textit{Procedure} & & & \textit{Bias} & \textit{s.d.} && \textit{Bias} & \textit{s.d.} \\
		\cline{1-2} \cline{4-5} \cline{7-8}
		\multirow{2}{*}{CAR}  & CAE      &  & 0.005  & 0.21 & & 0.011  & 0.21 \\
		& CE       &  & 0.002  & 0.43 & & -0.292 & 0.44 \\
		\cline{1-2} \cline{4-5} \cline{7-8}
		\multirow{2}{*}{CRC} & CAE     &   & -0.032 & 0.54 & & -0.024 & 0.53 \\
		& CE      &   & -0.037 & 0.78& & -0.313 & 0.75 \\
		\cline{1-2} \cline{4-5} \cline{7-8}
		CRU           & CE      &   & -0.003 & 0.08 && -2.281 & 0.09 \\
		\hline
	\end{tabular}

\end{table}
\vspace{-0.2in}

\section{Conclusion}
\label{Section:Conclusion}

\noindent

In this paper, we propose a cluster-adaptive network A/B testing procedure which consists of a cluster-adaptive randomization and a cluster-adjusted estimator. Utilizing  the cluster features to produce balanced treatment arms of  clusters, the proposed cluster-adaptive randomization is able to improve the estimation efficiency by reducing the variance of the ATE estimator. The cluster-adjusted estimator, which  eliminates the network effect gotten from linked clusters and adjusts the inflated cluster effect,   achieves a consistent estimation for the ATE,  when the employed randomization satisfies the balance condition (\ref{ass1}). Numerical studies suggest that our cluster-adaptive network A/B testing procedure outperforms other methods, as it can produce consistent estimation  and reduce the variance of the estimator. These results offer  a new angle to study the network A/B testing problem. The randomization steps should take the balance of the cluster features into account, so that the treatment arms are more comparable. Moreover, the balanced treatment arms may improve the estimation of the ATE. These ideas  are  valuable  in both theoretical and applied aspects, and  can help in the design and analysis of  novel network A/B testing  procedures.

\section{Broader Impact}

Any social network based on data driven methods are in  potential risk of privacy leakage and ethical problems. Though the usage of network information can improve the decision-making process, researchers have to be aware of some other consequences caused by their algorithm.  In this paper,  features of clusters in a social network are utilized to improve the results of network A/B testing. In practice, ethical and  privacy concerns should be taken into account for  the choice of the features put into the algorithm. It can be seen from the simulation results that even if not all of the features are put into the algorithm, the bias and efficiency can still be improved a lot. Therefore, if the sensitive features, which may introduce potential risks, are not used, the cluster adaptive network A/B testing procedure still works well. Therefore, some of the potential risks can be avoided by the   appropriate choice of the  features used in the algorithm. The proposed procedure can be implemented in many fields  comprising drug development, economics, E-commerce, finance, network data analysis and online testing.

\bibliographystyle{plain}
\bibliography{reference}

\newpage
\appendix
\section{Proof of the Main Results}
\label{Section: proofs}

\subsection{Proof of Theorem 2.1.}

Theorem~\ref{Th2} characterizes the balance property of CAR. It shows that as the number of clusters increases the Mahalanobis distance converges to zero in probability, so that the standardized difference of the mean for the clusters' covariates is minimized as well. As a result, the treatment arms are more feasibly compared.

\begin{proof}[Proof of Theorem~\ref{Th2}]
 Similar to the approach in  \cite{qin2016pairwise}, this proof consists of the following two main steps.  

\begin{enumerate}
    \item First, let $\Sigma= \text{Cov}[\bX_j]>0$, and $\xi_j = \Sigma^{-1/2}\left(\bX_j-\mathbb{E}[\bX_j]\right)$, where  $\Sigma^{-1/2}$ is the Cholesky square root of $\Sigma^{-1}$. Therefore, $\mathbb{E}[\bm{\xi}_j]=0$ and $\ocov[\bm{\xi}_j]=\mathbf{I}$. Consider 
    
    \begin{align}
        \tilde{M}_m & = \frac{m}{4}\left(\bar{\bX}_A-\bar{\bX}_B \right)^{\text{t}}\Sigma^{-1}\left(\bar{\bX}_A-\bar{\bX}_B \right) \nonumber\\
        &= \frac{m}{4}\left(\bar{\bm{\xi}}_A-\bar{\bm{\xi}}_B\right)^{\text{t}}\left(\bar{\bm{\xi}}_A-\bar{\bm{\xi}}_B\right),  \label{proof.lm1.eq0}
    \end{align}
where  $\bar{\bm{\xi}}_A$ and $\bar{\bm{\xi}}_B$ are the sample averages of $\bm{\xi}_j$ for treatments $A$ and $B$, respectively. Using the drift condition \cite{meyn2012markov}, we prove that $\tilde{M}_m=O_P(m^{-1})$    in Lemma~\ref{lm1}. 
    \item  As $m \to \infty$, $\ocov(\bX) \stackrel{\mathscr{P}}{\longrightarrow}\Sigma$, then  it follows from Lemma~\ref{lm1}, the continuous mapping theorem, and the  Slutsky's theorem, that (\ref{th2.eq1}) in Theorem 2.1 follows. 
\end{enumerate}
\end{proof}

The key ingredient for  the proof of Theorem~\ref{Th2}   is presented in the following lemma. 
\begin{lemma}
\label{lm1}
Under the condition of Theorem~\ref{Th2}, $\lbrace m \tilde{M}_{m} \rbrace_{m=1}^{\infty}$ is a sequence of Harris recurrent Markov chains, and   $\tilde{M}_m = O_p(m^{-1})$.

\end{lemma}

\begin{proof}[Proof of Lemma~\ref{lm1}.]
By (\ref{proof.lm1.eq0}), it can be seen  that
\begin{align*}
    \tilde{M}_m = \frac{1}{m}\bm{\eta}_m^{\text{t}} \bm{\eta} _m,
\end{align*}
where $\bm{\eta}_m= \frac{m}{2}\left(\bar{\bm{\xi}}_A-\bar{\bm{\xi}}_B\right) =\sum_{T_j=1}\bm{\xi}_j -   \sum_{T_j=0}\bm{\xi}_j$. Then if we can show $\bm{\eta}_m$ follows a stationary distribution, then it follows that $\bm{\eta}_m^{\text{t}}\bm{\eta}_m = O_p(1)$. Therefore, Lemma~\ref{lm1} follows from Slutsky's theorem.

In order to verify that $\bm{\eta}_m$ follows a stationary distribution, consider  verifying the drift condition, i.e., (iii)  of Theorem 11.0.1. in \cite{meyn2012markov}. Define $\bm{V}(\bm{\eta}_m)=\bm{\eta}_m^{\text{t}}\bm{\eta}_m$ as the test function and 
\begin{align*}
\bm{\Delta}_{m+2}&=\bm{\eta}_{m+2}- \bm{\eta}_{m}\\
&  = (-1)^{T_{m+2}}(\xi_{m+1}-\xi_{m+2}),
\end{align*}
where $\lbrace \bm{\eta}_{2k} \rbrace_{k=1}^\infty$ is a Markov process. To calculate  the test function, let $\mathbb{E}[\cdot|\bm{\eta}_m]=\mathbb{E}_m[\cdot]$. It follows that
\begin{align}
\mathbb{E}_m[\bm{V}(\bm{\eta}_{m+2})] -\bm{V}(\bm{\eta}_m)& = 2\bm{\bm{\eta}}_{m}^{\text{t}} \mathbb{E}_m\left[\bm{\Delta}_{m+2}\right]  + \mathbb{E}_m\left[\bm{\Delta}_{m+2}^{\text{t}}\bm{\Delta}_{m+2}\right]  , \label{p.lm1.eq1}
\end{align}
where $\mathbb{E}_m [\bm{\Delta}_{m+2}^{\text{t}}\bm{\Delta}_{m+2}]=\mathbb{E}_n\left[(-1)^{2T_{m+2}}(\bm{\xi}_{n+1}-\bm{\xi}_{n+2})^{\text{t}}(\bm{\xi}_{n+1}-\bm{\xi}_{n+2})\right]\ge0$. For the first term in (\ref{p.lm1.eq1}), it follows that
\begin{align*}
\bm{\eta}_{m} ^{\text{t}}   \mathbb{E}_m[\bm{\Delta}_{m+2}] & = \mathbb{E}_m \left[\mathbb{E} \left[(-1)^{T_{m+2}}\bm{\eta}_m^{\text{t}} (\bm{\xi}_{m+1}-\bm{\xi}_{m+2})  |\bm{\eta}_m, \bm{\xi}_{m+1} ,\bm{\xi}_{m+2} \right]\right]\\
& =\mathbb{E}_{m}\left[ (1-2q)  \left|\bm{\eta}_m^{\text{t}}( \bm{\xi}_{m+1}-\bm{\xi}_{m+2})\right|\right]\\
&= (1-2q)\left|\bm{\eta}_m\right| \cdot\mathbb{E}_m\left[ \left|\bm{\xi}_{m+1}-\bm{\xi}_{m+2}\right|\right]\cdot \mathbb{E}\left[|cos\theta|\right],
\end{align*}
where $\theta$ is the angle between $\bm{\eta}_m$ and $\bm{\xi}_{m+1}-\bm{m+2}$. Here, $\mathbb{E}_m\left[ \left|\bm{\xi}_{m+1}-\bm{\xi}_{m+2}\right|\right]$ and $ \mathbb{E}\left[|cos\theta|\right]$ are two positive constants.  It follows from $q>1/2$ that there exist constants $b>0,c>0$ such that 
\begin{equation*}
 \begin{array}{lll}
       \mathbb{E}_m\left[ \bm{V}(\bm{\eta}_{m+2})\right]- \bm{V}(\bm{\eta}_m)< -c     & \textit{for} \ \   |\bm{\eta}_{m}|>b  & and \\
     \mathbb{E}_m\left[ \bm{V}(\bm{\eta}_{m+2})\right]- \bm{V}(\bm{\eta}_m)<  \mathbb{E}_m\left[ \bm{\Delta}_{m+2}^{\text{t}}\bm{\Delta}_{m+2} \right]     & \textit{ for} \ \  |\bm{\eta}_{m}|\le b. 
    \end{array}
\end{equation*}
Therefore, the drift condition is checked and the proof of  the lemma is completed.

\end{proof}

\subsection{Proof of Theorem 2.2.}

The representation of CAE, i.e., (\ref{th1.presentation}), not only shows that CAE can eliminate the network effect, but also provides a way to study the asymptotic behavior of CAE. Theorem~\ref{Th1} only studies the consistency with the balance condition (\ref{ass1}). If more assumptions of network and the randomization procedure are provided, the asymptotic normality could be derived via (\ref{th1.presentation}). The proof of Theorem~\ref{Th1} is shown  as follows.

\begin{proof}[Proof of Theorem~\ref{Th1}.]
 Under the assumption of Theorem~\ref{Th1}, we rewrite the response model as follows.  For  $i\in\mathcal{C}_j,$
	\begin{align*}
	 Y_i 
	& =  \mu_0(1-Z_i) + \mu_1 Z_i + \bm{\beta}\bm{X}_j  + \sum_{k\in\mathcal{C}_j}\gamma_{ki}+ \sum_{k \notin\mathcal{C}_j}\gamma_{ki}+ \epsilon_i,
	\end{align*}
	where $\gamma_{ki}= \alpha_0A_{ik}Z_k(Z_i-1)+\alpha_1A_{ki}(Z_k-1)Z_i$ and $A_{ik}$ is the component of row $i$ column $k$ for  the adjacent matrix $A$. If the employed randomization is performed on clusters, then $\sum_{k\in\mathcal{C}_j}\gamma_{ki}=0$ as $Z_{k}=Z_i$, for $i,k\in\mathcal{C}_j$. If the users used for estimation are  all from the uncontaminated set, i.e., $\mU =  \{i: \bA_{i*} (\boldsymbol{1}-\bZ) = 0 \text{ and } Z_i =1\} \cup \{i: \bA_{i*} \bZ = 0  \text{ and } Z_i =0\},$
	then $\sum_{k\notin \mathcal{C}_j}\gamma_{ki}=0$ for $i\in\mU$.	Therefore,  it follows that 
\begin{align}
    Y_{i}=\mu_0(1-Z_i) + \mu_1 Z_i + \bm{\beta}^{\text{t}}\bm{X}_j  + \epsilon_i. \label{proof.th2.eq1}
\end{align}
It can be seen from (\ref{proof.th2.eq1}) that the network effect is  eliminated and (\ref{th1.presentation}) in Theorem~\ref{Th1}  follows. 

Now we  show that $\hat{\tau}_{CAE}$ is consistent. Note that  the users' cluster features will have  the same  value if the users are in the same cluster, so 
	\begin{align*}
	\bar{W}_j(\mathcal{U}) & =  \frac{1}{n_j(\mathcal{U})} \sum_{i\in \mathcal{C}_j\cap \mathcal{U}} (\bm{X}_j ^{\text{t}}\bm{\beta}+ \epsilon_i)\\
	& = \bm{X}_j^{\text{t}} \bm{\beta}_j + \frac{1}{n_j(\mathcal{U})} \sum_{i\in \mathcal{C}_j\cap \mathcal{U}}  \epsilon_i.
	\end{align*}
	According to assumption (\ref{ass1}) in Theorem~\ref{Th1}, it is easy to see that
	\begin{align*}
	 \begin{array}{cccc}
	 	\frac{m_A}{m}\stackrel{\mathscr{P}}{\longrightarrow}1/2&   \frac{m_b}{m}\stackrel{\mathscr{P}}{\longrightarrow}1/2 &  and &\frac{1}{m} \sum_{j=1}^m (2T_j-1)\bm{X}_j ^{\text{t}}\bm{\beta}\stackrel{\mathscr{P}}{\longrightarrow}0.
	 \end{array}
	\end{align*}
	Then the consistency of $\hat{\tau}_{CAE}$ follows from Slutsky's theorem and
	\begin{align}
	    \frac{1}{m} \sum_{j=1}^m  \frac{2T_j-1}{n_j(\mathcal{U})} \sum_{i\in \mathcal{C}_j \cap\mathcal{U}} \epsilon_i\stackrel{\mathscr{P}}{\longrightarrow}0,
	\end{align}
which can be shown by Chebyshev's inequality, provided that  $\epsilon_ i$ are i.i.d. with mean zero and variance $\sigma_{\epsilon}^2<\infty$.

\end{proof}

\subsection{Proof of Corollary 2.1.}

\begin{proof}[Proof of Corollary 2.1.]

\begin{enumerate}
    \item  Under CRC, $T_j \stackrel{i.i.d.}{\sim} Bernoulli (1/2)$ and $\bm{X}_j \perp T_j$, then
    \begin{align*}
    \begin{array}{ccc}
    \mathbb{E}[2T_j-1]=0 & and & \mathbb{E}[(2T_j-1)\bm{X}_j]=0. 
    \end{array}
    \end{align*}
    Therefore,  (\ref{ass1}) in Theorem~\ref{Th1} follows from WLLN. 
    
    \item Under CAR, assume $m$ is even, then $m_A=m_B = m/2$ a.s. and  $2\sum_{j=1}^m T_j-1  = m_{A}-m_{B}=0$. According to the proof of Lemma 3.1., $\bm{\eta}_m$ follows a stationary distribution. It follows that 
    \begin{align*}
        \bar{\bX}_A-\bar{\bX}_B & = \frac{2}{m} (\bar{\bm{\xi}}_A-\bar{\bm{\xi}}_B)= \frac{2}{m}\cdot\bm{\eta}_m\\
        & = O_p\left(\frac{1}{m}\right),
    \end{align*}
Therefore, the balance condition in Theorem 2.2.  follows.
\end{enumerate}

The balance condition is checked for CRC and CAR, hence the proof of Corollary 2.1 is completed.
	\end{proof}

	\vspace{0.2in}
	
\subsection{Proof of Theorem 2.3.}

Theorem~\ref{Th3} shows the benefit of using CAR to  produce balanced treatment arms in the estimation of the  ATE. As the lower bound in (\ref{lb.eq}) is positive, using CAR following with CAE will be more efficient than using CRC following with CAE. This result can further imply that the valid statistical inference for the ATE, i.e., the test that obtains the correct type I error,   using our cluster-adaptive network A/B testing, can achieve higher power and tighter confidence intervals. These results are left for future studies. The proof of Theorem~\ref{Th3} is as follows.

\begin{proof}[Proof of Theorem~\ref{Th3}.]

	Suppose the pseudo response variables $Y_j^P$ on the cluster level are generated by the same additive linear model
	$$Y^P_j = \mu_0 (1-T_j) + \mu_1 T_j + \bbeta^{\text{t}}\bX_j + \epsilon_j,$$
	where $\epsilon_j$ has the same distribution as $\epsilon_i$ in the response model for each user.  Then the  ATE based on pseudo response variables is
	
	\begin{align*}
	\hat\tau^C & =  \frac{1}{m_A}\sum_{j=1}^m T_jY_j^P  - \frac{1 }{m_B}\sum_{j=1}^m (1-T_j)Y_j^P  \\
	& \simeq  \tau + \bbeta^{\text{t}}(\overline\bX_A- \overline\bX_B) + \frac{2}{m}\sum_{j=1}^m (2T_j-1)\epsilon_j,
	\end{align*}
	where $\bar{\bm{X}}_A $, $\bar{\bm{X}}_B$ are  the sample means of cluster covariates in treatments A and B, respectively. Here $\hat{\tau}^C$ is a pseudo treatment effect estimator, which assumes that each cluster has only one response.\\

	The adjusted cluster-level estimator $\hat\tau_{CAE}$ only uses the sample average of the  $n_j(\mathcal{U})$ users in cluster $\mC_j$, i.e.,
	$$ \overline Y_{j}(\mathcal{U})  = \mu_0 (1-T_j) + \mu_1 T_j + \bbeta^{\text{t}}\bX_j + \frac{1}{n_j(\mathcal{U})}\sum_{i\in \mathcal{C}_j \cap\mathcal{U}}\epsilon_i.$$
	Then  CAE can be written as 
	\begin{align*}
	\hat\tau_{CAE} & =\frac{1}{m_A}\sum_{j=1}^m T_j \bar{Y}_j(\mathcal{U}) - \frac{1}{m_B} \sum_{j=1}^m( 1-T_j) \bar{Y}_j(\mathcal{U})\\
	& \simeq \tau +  \bbeta^{\text{t}}(\overline \bX_A - \overline \bX_B)  + \frac{2}{m} \sum_{j=1}^m  \frac{2T_j-1}{n_j(\mathcal{U})} \sum_{i\in \mathcal{C}_j \cap\mathcal{U}} \epsilon_i,\\
	where~	&\ovar \left[ \frac{2}{m} \sum_{j=1}^m  \frac{2T_j-1}{n_j(\mathcal{U})} \sum_{i\in \mathcal{C}_j \cap\mathcal{U}} \epsilon_i \right]  = \frac{4}{m^2}\sum_{j=1}^m \frac{\sigma_\epsilon^2}{n_j (\mathcal{U})} .
	\end{align*}
Denote $*$ as one of CAR or CRC. Then, the variances for two types of estimators under a  given  randomization * are
	\begin{align*}
	\ovar[\hat\tau^C|*] & \simeq \bbeta^{\text{t}}\ocov[\overline\bX_A- \overline\bX_B|*]\bbeta + \ovar\left[\frac{2}{m}\sum_{j=1}^m (2T_j-1)\epsilon_j \right]\\
	& =  \bbeta^{\text{t}}\ocov[\overline\bX_A- \overline\bX_B|*] \bbeta + \frac{4}{m} \sigma_\epsilon^2,\ \ \ \  and\\
	\ovar[\hat\tau_{CAE}|*] & \simeq \bbeta^{\text{t}}\ocov[\overline\bX_A- \overline\bX_B|*]\bbeta + \frac{4}{m^2}\sum_{j=1}^m \frac{\sigma_\epsilon^2}{n_j (\mathcal{U})}\\
	& \leq \bbeta^{\text{t}}\ocov[\overline\bX_A- \overline\bX_B|*] \bbeta + \frac{4}{m} \frac{1}{ \min_{j} n_j(\mathcal{U})} \sigma_\epsilon^2\\
	& \leq\ovar[\hat\tau^C|*].
	\end{align*}
	Then, under any randomization procedure $*$ we have
	\begin{align*}
	\ovar[\hat\tau_{CAE}|*] - \ovar[\hat\tau^C|*] & =  \frac{4}{m}\left[ \frac{1}{m}\cdot\sum_{j=1}^m \frac{1}{n_j(\mathcal{U})}-1\right] \sigma_\epsilon^2 := C_1 \sigma_\epsilon^2 ,
	\end{align*}
	where $C_1$ is some  constant depending on $\mathcal{G}$. As $n_j(\mathcal{U})\geq1$, it can be seen that $C_1\le0$ ($C_1 =0$ when $n_{j}(\mathcal{U})=1$ for  $j=1,\cdots,m$). By Theorem 3.2. in \cite{qin2016pairwise}, it follows   that $$PRIV(\tau^C|CAR) = \frac{\ovar[\hat\tau^C|CR] - \ovar[\hat\tau^C|CAR]}{\ovar[\hat\tau^C|CR]} = (1-\frac{\mE[M_m|CAR]}{p})R_C^2,$$
	
	and hence
	
	\begin{align*}
	PRIV(\hat\tau_{CAE}|CAR) & = \frac{\ovar[\hat\tau_{CAE}|CRC] - \ovar[\hat\tau_{CAE}|CAR]}{\ovar[\hat\tau_{CAE}|CRC]} \\
	& = \frac{(\ovar[\hat\tau^C|CRC] + C_1 \sigma_e^2 )- (\ovar[\hat\tau^C|CAR] + C_1 \sigma_e^2) }{\ovar[\hat\tau_{CAE}|CRC]}\\
	& = (1-\frac{\mE[M_m|CAR]}{p})\cdot R_C^2\cdot \frac{\ovar[\hat\tau^C|CR]}{\ovar[\hat\tau_{CAE}|CRC]}\\
	& \geq (1-\frac{\mE[M_m|CAR]}{p})R_C^2.
	\end{align*}
This completes the proof of the theorem.
\end{proof}

\section{The Hypothetical Network}
\label{Section: Hypothetical network append}
\vspace{-0.05in}

In this section, we study the properties of our cluster-adaptive network A/B testing procedure under different structures of the hypothetical network. Specially, we are interested in the performance of our procedure when the number of edges connecting different clusters increases.

We use the same hypothetical network as in Section 3. First, 500 independent clusters are generated by  the  Watts-Strogatz small-world model. For $j = 1,\cdots,500$, we generate  the cluster size $n_j$ from a symmetric discrete distribution,  $\mP(n_j = \lambda) = \left[\left(|\lambda- 20|+ 0.5\right)\sum_{\lambda=10}^{30}(1/ \left(|\lambda- 20|+ 0.5\right)\right]^{-1}$, where $\lambda=10,...,30$. Therefore, $N = \sum_{j=1}^{500}n_j$ users are involved in this network.  In addition, $r\times N$ edges are randomly added, so that  several clusters are connected, where  $r$ is chosen as the re-connection probability.

Consider the response $Y_i$ follows the assumed response model (\ref{Y.eqsd}) with $p=4$ and $\epsilon_i\sim N(0, 2^2)$.  The covariates of the clusters,  i.e., $\bX_j= (X_{j1},\cdots,X_{j4})^{\text{t}}$, involved in the response model are  the  number of vertices in $\mC_j$ ($X_{j1}$), the number of edges in $\mC_j$ ($X_{j2}$), the  number of edges in $\mC_j$ with other clusters ($X_{j3}$), and the density of $\mC_j$ ($X_{j4}$). In addition, consider $\alpha_1=-\alpha_0=\alpha=0.5$ for the network spill-over effect, $\beta_1= \cdots= \beta_4= 1$ for the effect of clusters' covariates, and $\tau=\mu_1-\mu_0 = 1$  for the treatment effect. 
\begin{equation}
\label{Y.eqsd} Y_i  =    \mu_0(1-Z_i) + \mu_1 Z_i + \alpha_0 \bA_{i*}\bZ (Z_i -1) + \alpha_1 \bA_{i*}(\bZ-\boldsymbol{1}) Z_i + \bbeta^{\text{t}}\bX_j + \epsilon_i, 
\end{equation}
To compare with other network A/B testing procedures, we consider  the classic A/B testing procedure, i.e., complete randomization with users (CRU),   following with the classical estimator (CE),  and  graph  clustering procedure using the  complete randomization on clusters (CRC), following with CE and CAE. For our CAR method, we consider two cases,  1)  CAR (2), CAR using the first  two  covariates   $X_{j1}$ and $X_{j2}$, and 2) CAR (4),  CAR using all of the four covariates $X_{j1},\cdots,X_{j4}$. The Mahalanobis distance and the standardized  difference in covariate means for $\bm{X}_j$ are presented in Figure~\ref{Figure1: Mahalanobis distance for the hypothetical network sup} and Figure~\ref{Figure2: standardized difference in means for the hypothetical network sup}  to compare the balance properties of the  randomization procedures. We compare the performance of these A/B testing procedures for the  hypothetical network with $r=0.1,0.2,\cdots,2.$

Table~\ref{Table 1: bias and s.d. for the hypothetical network sup}  summarizes the bias and standard deviation (s.d.) of the estimated ATE for $r=0.5,1,1.5$ and 2. Furthermore, we study the  $PRIV(\hat{\tau}_{CAE}|CAR)$  and its' corresponding lower bound (L-B)  in Table~\ref{Table 2: PRIV for the hypothetical network sup} for $r=0.5,1,1.5,$ and $2$. Figure~\ref{Figure 3: bias and sd for the hypothetical netwrok sup} summarizes the bias and the standard deviation of the estimated ATE, and the  $PRIV(\hat{\tau}_{CAE}|CAR)$  and its corresponding lower bound for all different settings of $r$. Note that the lower bound of  $PRIV(\hat{\tau}_{CAE}|CAR)$  for CAR (4), i.e., L-B (4), is calculated from (\ref{lb.eq}) in Theorem 2.3 with the Mahalanobis distance using  all four covariates and $R^2_C$ calculated according to the pseudo response model (\ref{psedo}) with all four cluster covariates, whereas the lower bound of  $PRIV(\hat{\tau}_{CAE}|CAR)$  for CAR (2), i.e., L-B (2), is calculated with the Mahalanobis distance using the first two covariates used in randomization and $R^2_C$ calculated according to the pseudo response model (\ref{psedo}) with only the first two covariates.

 \textbf{Results}:  Figure~\ref{Figure1: Mahalanobis distance for the hypothetical network sup} and Figure~\ref{Figure2: standardized difference in means for the hypothetical network sup} show that CAR produces better balance for the clusters' covariates  than CRC under different network settings. In particular, CAR (4) outperforms other  procedures for all values of $r$. It can be seen from Figure~\ref{Figure2: standardized difference in means for the hypothetical network sup} that the standard deviation for each component of $\bm{X}_j$ are the smallest under  CAR (4). The two covariates used in CAR (2) have small standard deviations as well, but the two covariates  not used in CAR (2) have slightly larger  standard deviations. In addition, Figure~\ref{Figure1: Mahalanobis distance for the hypothetical network sup} shows that the Mahalanobis distances calculated with all four covariates are centered at zero for CAR(4) for all settings of $r$.  The values under CAR (2) are less centered than the corresponding values under CAR(4), but are more centered than the values under CRC.  These results show that  the standardized difference in means for the clusters' covariates  and the Mahalanobis distance  are much more centered around zero under CAR. As a result, the differences of  distributions of the clusters' covariates between treatment A and treatment B are much smaller under CAR, and the two treatment groups are more feasibly compared.

\begin{figure}[H]
	\centering
	\includegraphics[width=1\linewidth]{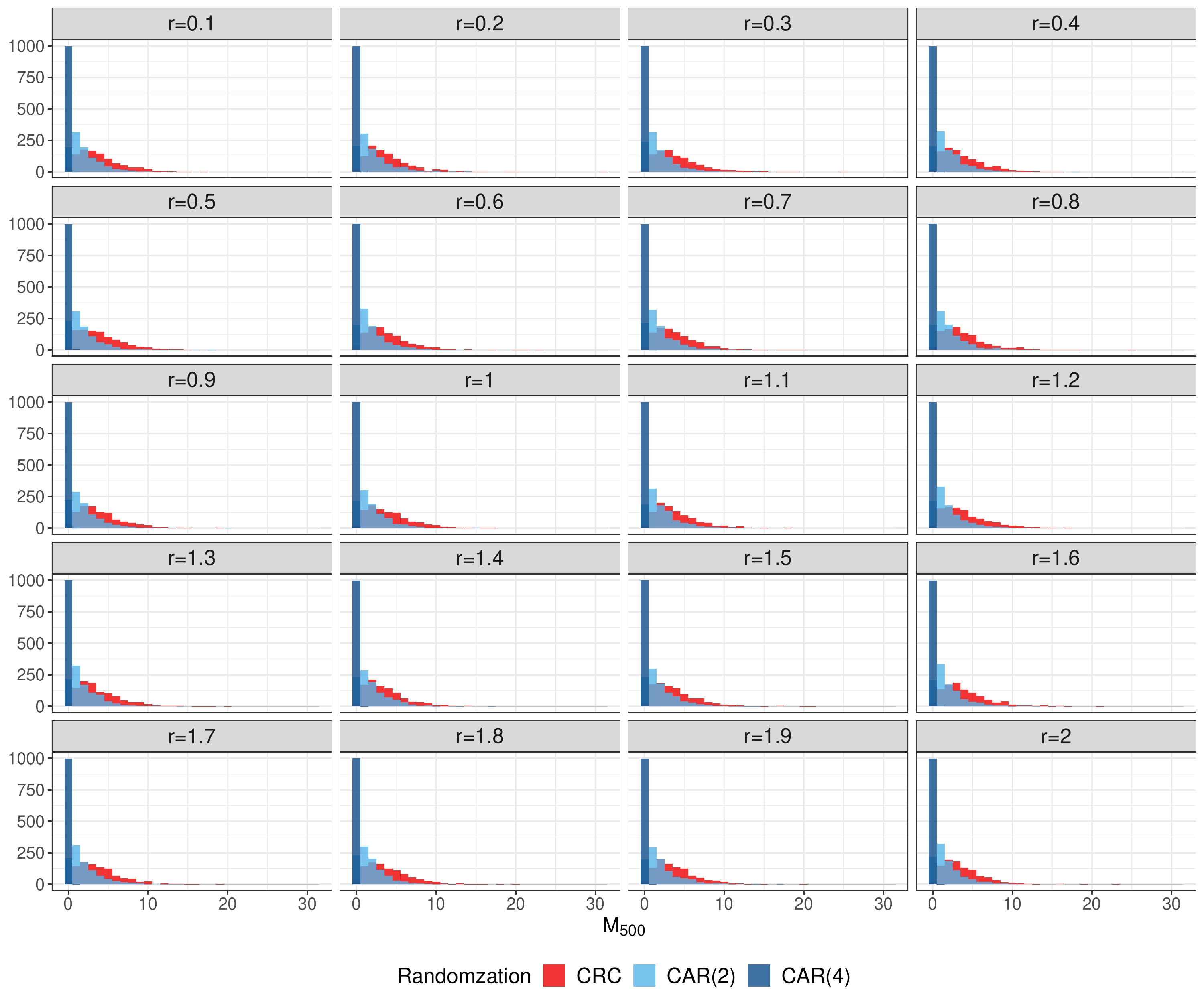}
		\caption{The Mahalanobis distance given the hypothetical network, based on 1000 runs. The Mahalanobis distance is calculated for all four covariates of the clusters.  }
		\label{Figure1: Mahalanobis distance for the hypothetical network sup} 
\end{figure}

Table~\ref{Table 1: bias and s.d. for the hypothetical network sup} and Figure~\ref{Figure 3: bias and sd for the hypothetical netwrok sup} show that the network A/B testing procedures using CAE, i.e., CAR following with CAE and CRC following with CAE, are unbiased for all values of $r$. On the other hand, the  bias of CE increases as $r$ increases under CAR, CRC, and CRU. It can be seen that  CE is less biased under CRC and CAR than the value under CRU. This is because CRC and CAR are performed on clusters, so that the network spill-over effect obtained from  the users within the same cluster is removed. However, the network spill-over effect obtained from the users from the linked clusters still exists. Therefore, the bias under CAR and CRC still increases as  $r$ increases.

\begin{figure}[H]
	\centering
	\includegraphics[width=1\linewidth]{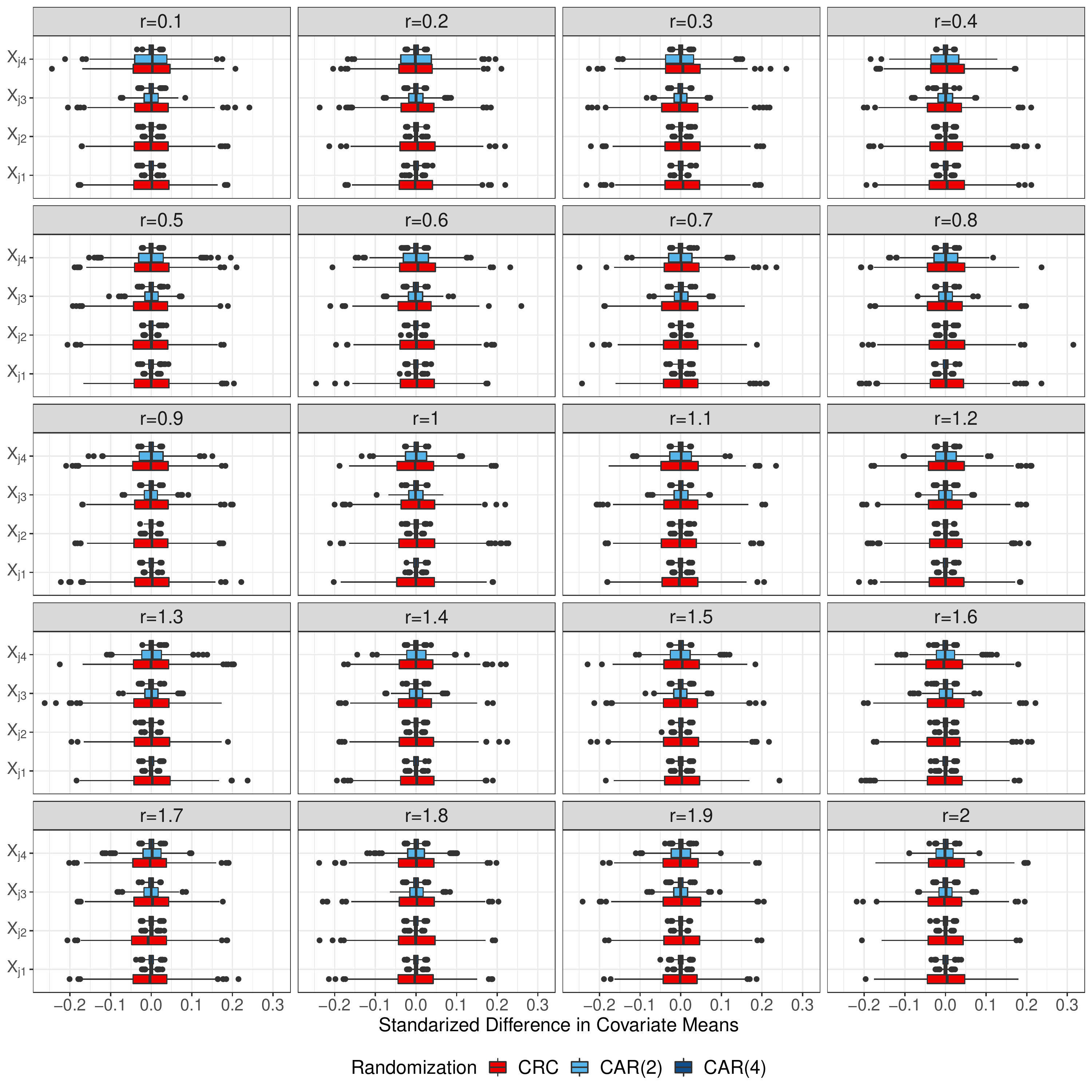}
		\caption{The standardized difference in means for  $\bX_j$, given the hypothetical network, based on 1000 runs. }
		\label{Figure2: standardized difference in means for the hypothetical network sup}
\end{figure}
To further compare the unbiased network A/B testing procedures, it can be seen from Table~\ref{Table 1: bias and s.d. for the hypothetical network sup} and Figure~\ref{Figure 3: bias and sd for the hypothetical netwrok sup} that using CAE under CAR (4) results in the smallest standard deviation among all of the unbiased procedures. The performance of  using CAE under CAR (2) is also better than the performance of using CAE under CRC. In addition, as the value of $r$ increases, the standard deviation of the CAE increases. This is because that, as $r$ increases, the number of edges connecting different clusters increases, and thus the number of users in the uncontaminated set decreases. When fewer users can be used for estimation, the efficiency of CAE  will decrease.  When $r=2$, which means that there are $2\times N$ edges among the clusters, the results still show that using CAR together with CAE is more efficient than CRC, as the PRIVs are about 70\% and 64\%, respectively.

\begin{table}[H]
\centering
	\caption{The bias and the standard deviation (s.d.) of the estimated  ATE under different A/B testing procedures given the hypothetical network,  based on 1000 runs.}
	\label{Table 1: bias and s.d. for the hypothetical network sup}
\resizebox{\textwidth}{!}{ 
\begin{tabular}{ccrrrrrrrrrrrrr}
\hline
\multirow{2}{*}{\textit{Randomization}} & \multirow{2}{*}{\textit{Estimation}} &  & \multicolumn{2}{c}{$r=0.5$} & & \multicolumn{2}{c}{$r=1$}& & \multicolumn{2}{c}{$r=1.5$} & & \multicolumn{2}{c}{$r=2$} \\
 \cline{4-5} \cline{7-8} \cline{10-11} \cline{13-14}
        &    &  & \textit{bias}     & \textit{s.d.}  &     & \textit{bias}     & \textit{s.d.}   &    & \textit{bias}     & \textit{s.d.}  &    & \textit{bias}     & \textit{s.d.}     \\
\cline{1-2}        \cline{4-5} \cline{7-8} \cline{10-11} \cline{13-14}
\multirow{2}{*}{CAR(2)} & CAE &      & 0.00  & 0.09 & & -0.00 & 0.10 & & -0.01 & 0.12& & -0.01 & 0.16 \\
        & CE   &    & -0.50 & 0.11 & & -1.00 & 0.10& & -1.50  & 0.10& & -2.00 & 0.09 \\
\cline{1-2}        \cline{4-5} \cline{7-8} \cline{10-11} \cline{13-14}
\multirow{2}{*}{CAR(4)} & CAE   &    & -0.00 & 0.06  & & 0.00  & 0.08 & & 0.00  & 0.11 & & 0.00   & 0.14 \\
        & CE    &  & -0.50 & 0.07& & -1.00 & 0.08 & & -1.50 & 0.08 & & -2.00 & 0.08 \\
\cline{1-2}        \cline{4-5} \cline{7-8} \cline{10-11} \cline{13-14}
\multirow{2}{*}{CRC} & CAE    &  & 0.00& 0.22 & & 0.01 & 0.24 & & 0.01  & 0.24& & -0.00  & 0.26 \\
        & CE     &     & -0.50 & 0.23 && -1.00 & 0.25& & -1.50 & 0.24 & & -2.00 & 0.24 \\
\cline{1-2}        \cline{4-5} \cline{7-8} \cline{10-11} \cline{13-14}
CRU     & CE      &  & -4.00 & 0.08 & & -4.50 & 0.08 & & -5.00 & 0.09 && -5.50 & 0.09\\
\hline
\end{tabular}}
\end{table}

\begin{figure}[H]
\begin{center}
	\includegraphics[width=1\linewidth]{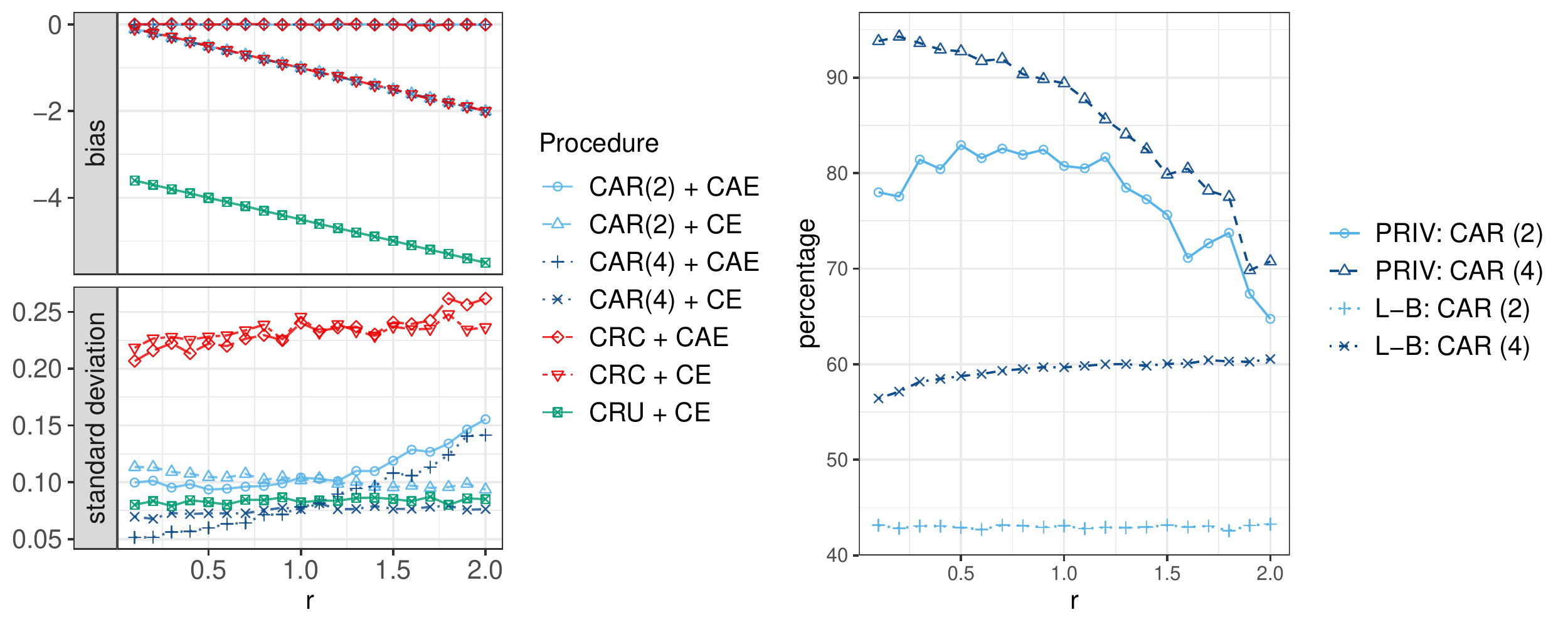}
\end{center}
		\caption{The bias and standard deviation of the estimated ATE, and the  $PRIV(\hat{\tau}_{CAE}|CAR)$    and the corresponding lower bound (L-B) versus the reconnecting probability $r$ given the hypothetical network, based on 1,000 runs.}
		\label{Figure 3: bias and sd for the hypothetical netwrok sup}
\end{figure}

\begin{table}[H]
\centering
\caption{The  $PRIV(\hat{\tau}_{CAE}|CAR)$  $(\%)$ and the corresponding lower bound (L-B) for the CAE under CAR , based on 1,000 runs. }
\label{Table 2: PRIV for the hypothetical network sup}
\begin{tabular}{ccccccccccccccccc}
\hline
\multirow{2}{*}{CAR} & & \multicolumn{2}{c}{$r=0.5$} & & \multicolumn{2}{c}{$r=1$} & & \multicolumn{2}{c}{$r=1.5$}  & & \multicolumn{2}{c}{$r=2$}   \\
\cline{3-4} \cline{ 6-7} \cline{9-10} \cline{12-13}
    & & PRIV   & (L-B) &   & PRIV   & (L-B) &    & PRIV  & (L-B)  &  & PRIV  & (L-B)     \\
    \hline
 (2) & & 82.93 & (42.91) & & 80.75  & (43.09) & & 75.63 & (43.16) && 64.75 & (43.26) \\
(4)  & & 92.75 & (58.76)  & & 89.41 & (59.67) & & 79.85 & (60.04) && 70.76 & (60.53)\\
\hline
\end{tabular}
\end{table}

Table~\ref{Table 2: PRIV for the hypothetical network sup}  and Figure~\ref{Figure 3: bias and sd for the hypothetical netwrok sup} show that the $PRIV(\hat{\tau}_{CAE}|CAR)$  for CAR(4) is getting closer to its corresponding lower bound, i.e., L-B CAR (4), as $r$ increases, whereas the lower bound of the $PRIV(\hat{\tau}_{CAE}|CAR)$  for CAR(2)  is quite conservative. The reason is  that the $R_c^2$ will be small when it is  calculated with only the first two covariates and  the true model (\ref{psedo}) includes  more covariates. The $PRIV(\hat{\tau}_{CAE}|CAR)$ under CAR (2) is above 64\%, and its corresponding lower bound is about 43\%.  The results suggest that when fewer covariates are used in CAR than the number of covariates that we should use, the efficiency of the estimation for the ATE still can be obtained.


\section{The MIT Phone Call Network}
\label{Section: MIT}

	The MIT phone call Network listed in the Network Data Repository   \cite{datarepository} is studied again for two purposes.  First,  the performance of CAR for the MIT phone call network is presented to demonstrate the advantages of using CAR for real world networks. Second, we compare the bias and  standard deviation of the estimated ATE for CAR using the first two clusters' covariates, i.e., CAR (2),  and CAR using all four covariates, i.e., CAR (4),  to show that if  more balanced treatment arms can be obtained from randomization, then we can obtain a more efficient estimation of the ATE.

	The response of the  users is generated with the same parameter  settings as  in Section~\ref{Section: Hypothetical network}. We consider $\alpha \in\lbrace  0,0.25,0.5,0.75, 1\rbrace$ for the settings of network spill-over effects. The  A/B testing procedures used in Section~\ref{Section: Hypothetical network} are also compared for the MIT phone call network.The balance properties of CAR (2), CAR (4), and CRC are also compared. For the simplicity of presentation, we only present the standardized difference in the covariate means for $\bm{X}_j$ and the Mahalanobis distance for $\alpha=0$ in Figure~\ref{Figure 4: CAR balance for the MIT netwrok sup}.  As the MIT network is fixed, the balance properties of CAR and CRC for other values of $\alpha$ are similar as presented in Figure~\ref{Figure 4: CAR balance for the MIT netwrok sup}. Therefore, these results are omitted. Figure~\ref{Figure 5: bias and sd for the MIT network sup} summarizes the bias and standard deviation of the estimated ATE for different network A/B testing procedures. 
	
\begin{figure}[H]
	\centering
	\includegraphics[width=1\linewidth]{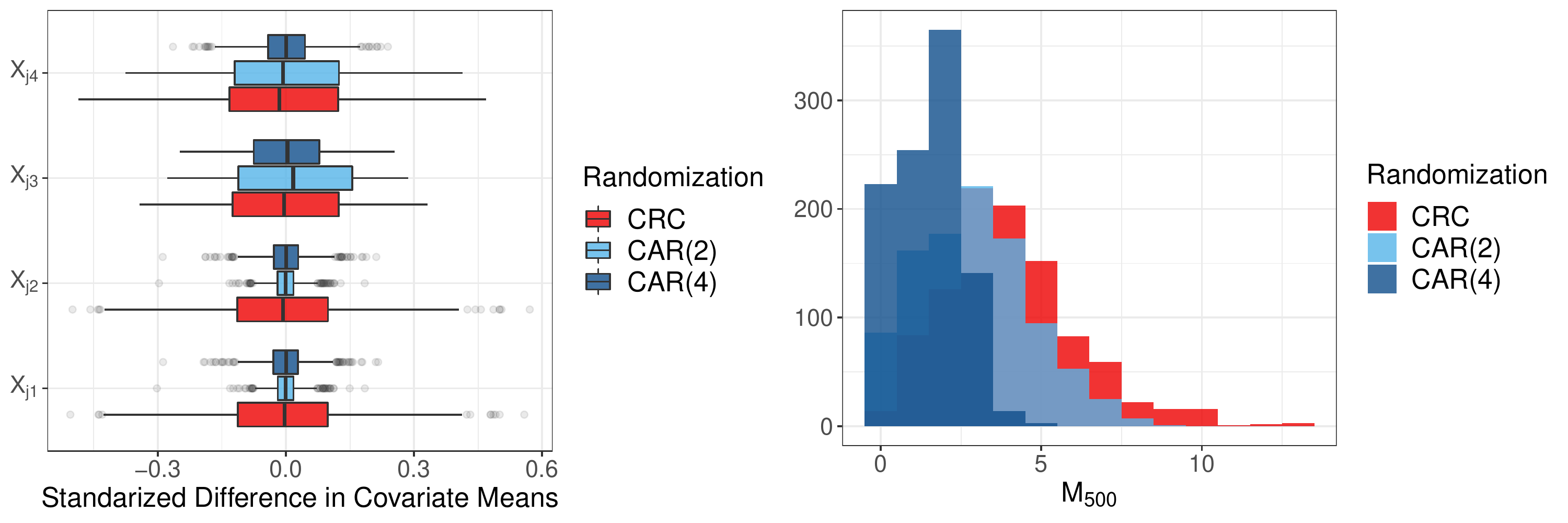}
		\caption{The standardized difference in the covariate means for $\bm{X}_j$ and the Mahalanobis distance given the MIT network ($\alpha = 0$), based on 1,000 runs.The Mahalanobis distance is calculated with all four  covariates of the clusters. }
		\label{Figure 4: CAR balance for the MIT netwrok sup}
\end{figure}

\begin{figure}[H]
	\centering
	\includegraphics[width=1\linewidth]{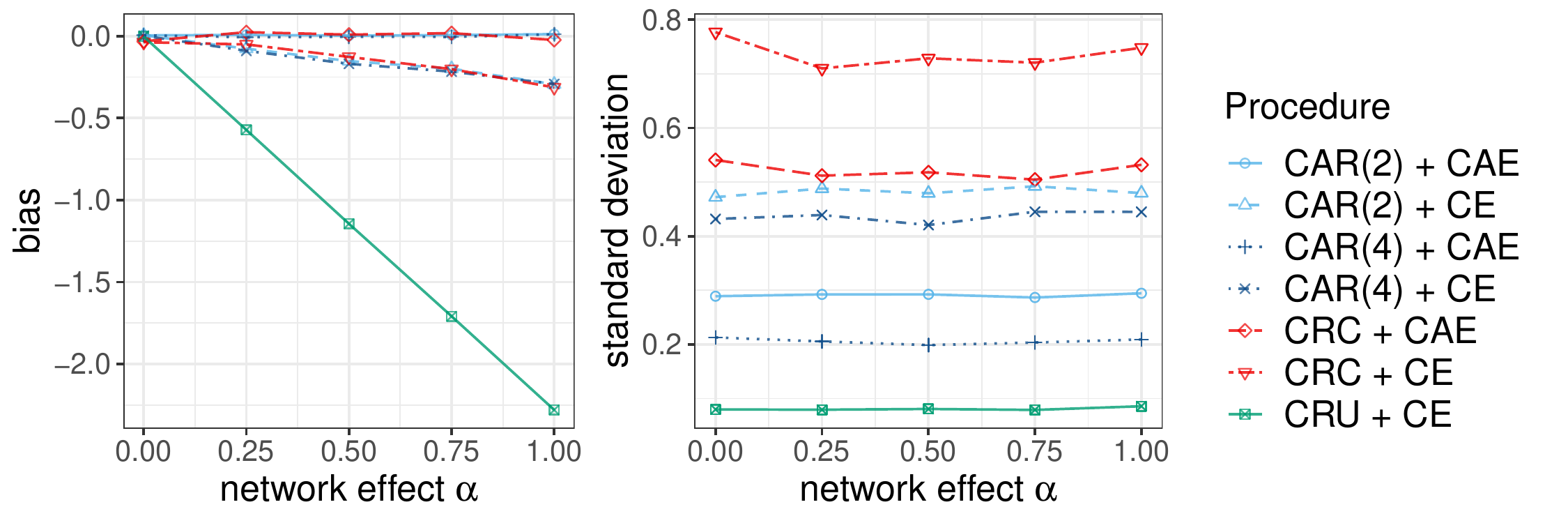}
		\caption{The bias and standard deviation of the estimated ATE versus the network effect $\alpha$ given the MIT  network, based on 1,000 runs.}
		\label{Figure 5: bias and sd for the MIT network sup}
\end{figure}

	\textbf{Results}:  
	Figure~\ref{Figure 4: CAR balance for the MIT netwrok sup} shows that the covariates of the clusters for the MIT Phone call network are more balanced between treatments A and B under CAR. It can be seen from the boxplots that the IQR of the covariates used in CAR are much smaller than the values under CRC. The Mahalanobis distance is much more centered under CAR than the value under CRC. These results illustrate the usage of CAR in improving the balance of  clusters' covariates for the MIT phone call network.
	
It can be seen from	Figure~\ref{Figure 4: CAR balance for the MIT netwrok sup}  and Figure~\ref{Figure 5: bias and sd for the MIT network sup} that using more covariates in CAR can result in more balanced treatment arms and hence more efficient estimation of the ATE.	Comparing CAR(2) and CAR(4) in Figure~\ref{Figure 4: CAR balance for the MIT netwrok sup}, as more covariates are used in CAR, the Mahalanobis distance and the  standardized difference in the covariates' means for $X_{j3}$ and $X_{j4}$ are more centered around zero. Therefore, the increase of the number of covariates used in CAR improves the balance of the clusters' covariates for  treatment A and treatment B. It can be seen from Figure~\ref{Figure 5: bias and sd for the MIT network sup} that as treatment A and treatment B are more comparable under CAR (4), the standard deviation of CAE  is much smaller  than the values under CAR (2).  

\section{Code}
See \url{https://github.com/yzhou12/CAR-network-AB-test} for the code and the  data used to run the experiments.

\end{document}